\documentclass[a4paper,USenglish,cleveref,autoref,thm-restate,nonumbering]{lipics-v2021}

\usepackage{amsthm,amsmath,amssymb}

 \usepackage{tcolorbox}

\newboolean{shortver}
\setboolean{shortver}{false}

\usepackage{graphicx} 

\usepackage{bm}
\usepackage[lined,boxed,commentsnumbered,ruled,vlined,noend,linesnumbered,boxed]{algorithm2e}
\usepackage{amsfonts}
\usepackage{tcolorbox}
\usepackage{mathtools}
\tcbuselibrary{skins}
\usepackage{xcolor}
\usepackage{cancel}
\usepackage{csquotes}
\usepackage[]{mdframed}
\usepackage{caption}
\usepackage{tcolorbox}
\tcbuselibrary{skins}
\usepackage[dvipsnames]{xcolor}
\usepackage{float}
\usepackage[bb=boondox]{mathalfa}
\usepackage{soul}

\colorlet{mix}{red!50!black}
\definecolor{DarkRed}{rgb}{0.8,0,0}
\definecolor{ForestGreen}{rgb}{0.1333,0.5451,0.1333}

\usepackage{hyperref}
\hypersetup{colorlinks={true},linkcolor={DarkRed},citecolor=ForestGreen}

\usepackage{tikz}
\usetikzlibrary{patterns}

\usepackage[colorinlistoftodos]{todonotes}

\usepackage{amsthm}
\usepackage{cleveref}

\newcommand{\nv}{\mathsf{NVis}}
\newcommand{\vis}{\mathsf{Vis}}
\newcommand{\bd}{\mathsf{\delta}}
\newcommand{\gc}{\mathsf{GCon}}

\newcommand{\vil}{\mathsf{N}}

\newcommand{\inte}{\mathsf{iN}}

\newcommand{\npo}{near polygon}

\newcommand{\pol}{\mathcal{P}}

\newcounter{sasanka}

\newcounter{satya}

\newtcolorbox{greenbox}[1][]{colback=green!60!white, colframe=green!50!black,
  boxrule=0.5pt, arc=2mm, left=2mm, right=2mm, top=1mm, bottom=1mm, #1}

\newcounter{debabrata}

\definecolor{bblue}{rgb}{0,0.1,0.55}

\definecolor{defblue}{rgb}{0.1,0.4,0.6} 
\let\emph\relax\DeclareTextFontCommand{\emph}{\color{bblue}\em}
 
\nolinenumbers

\title{{\sc Witness Set}: A Visibility Problem in ${\sf NP} \cap {\sf XP}$}

\titlerunning{{\sc Witness Set}: A Visibility Problem in ${\sf NP} \cap {\sf XP}$} 

\author{Satyabrata Jana}{Indian Statistical Institute, Kolkata, India}{satyamtma@gmail.com}{}{}

 \author{Debabrata Pal}
{Indian Statistical Institute, Kolkata, India}{debabratapal4521@gmail.com}{}{}

 \author{Bodhayan Roy}
{Indian Institute of Technology Kharagpur, Kharagpur, India}{bodhayan.roy@gmail.com}{}{}

 \author{Sasanka Roy}
{Indian Statistical Institute, Kolkata, India}{sasanka.ro@gmail.com}{}{}

\authorrunning{S. Jana, D. Pal, S. Roy, and B. Roy} 

\Copyright{Jane Open Access and Joan R. Public} 



\EventEditors{John Q. Open and Joan R. Access}
\EventNoEds{2}
\EventLongTitle{42nd Conference on Very Important Topics (CVIT 2016)}
\EventShortTitle{CVIT 2016}
\EventAcronym{CVIT}
\EventYear{2016}
\EventDate{December 24--27, 2016}
\EventLocation{Little Whinging, United Kingdom}
\EventLogo{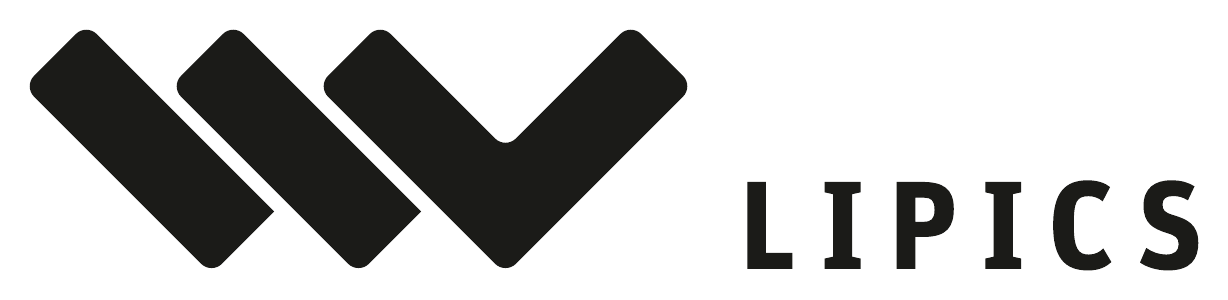}
\SeriesVolume{42}
\ArticleNo{23}

\begin{document}
 \maketitle

\begin{abstract}

We study the \textsc{Witness Set} problem, a natural dual\footnote{The \textsc{Art Gallery} problem is a covering problem, whereas the \textsc{Witness Set} problem can be viewed as a packing problem based on visibility regions. Moreover, the size of a maximum witness set provides a lower bound on the size of a minimum guard set.} to the classical \textsc{Art Gallery} problem. In the \textsc{Witness Set} problem, we are given a polygon $\mathcal{P}$ and an integer $k$ as input, and the objective is to determine whether $\mathcal{P}$ has a witness set of size at least $k$. A point set $X \subseteq \mathcal{P}$ is called a witness set if every point in $\mathcal{P}$ is visible from at most one point in $X$.

For simple polygons, we show that \textsc{Witness Set} lies in both $\mathsf{NP}$ and $\mathsf{XP}$. This stands in sharp contrast to its dual, the \textsc{Art Gallery} problem, which was recently shown to be $\exists\mathbb{R}$-complete by Abrahamsen et al.~\cite{DBLP:journals/jacm/AbrahamsenAM22}, and is therefore neither in {\sf NP} nor admits a polynomial-size discretization unless $\mathsf{NP} = \exists\mathbb{R}$. In contrast, we prove that \textsc{Witness Set} for simple polygons admits a finite discretization of size $n^{f(k)}$ for some function $f$.

For comparison, even for simple polygons, Efrat and Har-Peled~\cite{DBLP:journals/ipl/EfratH06} gave an algorithm for \textsc{Art Gallery} running in time $n^{\mathcal{O}(k)}$ using tools from real algebraic geometry, and it appears difficult to obtain such algorithms without this machinery (as discussed in~\cite{DBLP:journals/jacm/AbrahamsenAM22}). On the other hand, our approach for \textsc{Witness Set} is purely combinatorial and relies on discretization, leading to an $n^{f(k)}$-time algorithm.

Although Amit et al.~\cite[Section~3.3]{AmitMitchellPacker} claimed more than fifteen years ago that \textsc{Witness Set} is $\mathsf{NP}$-hard, no proof or reference was provided. We show that the discrete version of the \textsc{Witness Set} problem -- where the witness set must be chosen from a given finite point set $Q$ (instead of allowing witnesses to be chosen anywhere in the polygon), referred to as \textsc{Discrete Witness Set} -- is $\mathsf{NP}$-complete, even when the input is restricted to rectilinear polygons with holes. However, for simple polygons, \textsc{Discrete Witness Set} admits a polynomial-time algorithm by Das et al.~\cite{das2025witnesssetmonotonepolygons}. Thus, it remains an open question whether the \textsc{Witness Set} problem is {\sf NP}-hard.

\end{abstract}

\section{Introduction}

In a  polygon $\pol$, two points $p,q\in \pol$ are said to be \emph{visible} to each other if the line segment $\overline{pq}$ lies entirely
within $\pol$.  A point set $W\subseteq \pol$ is called a \emph{witness set} if no point
of the polygon is visible from more than one point of $W$ (see   Figure~\ref{witfirst} for an illustration). In contrast to
guarding problems, where the aim is to make every point of $\pol$ visible from
at least one guard, a witness set enforces exclusivity of visibility:
each location in the polygon ``corresponds'' to at most one witness.  A witness
set is {\em optimal} if it is of the largest possible size.  In
this work, we study the  decision version of the problem.

\begin{tcolorbox}[enhanced,
  title={\color{black}{\sc Witness Set}},
  colback=white, boxrule=0.4pt,
  attach boxed title to top center={xshift=-5cm,yshift*=-2.5mm},
  boxed title style={size=small,frame hidden,colback=white}]
  \vspace{-1mm}
  \textbf{Input:} A polygon $\pol$ and a non-negative integer $k$.\\
  \textbf{Question:} Does there exist a set $S$ of at least $k$ points in $\pol$ 
such that every point of the polygon $\pol$ is visible from at most one point of $S$?
\end{tcolorbox}


\begin{figure}[ht!]
    \centering
    \includegraphics[width=0.5\linewidth]{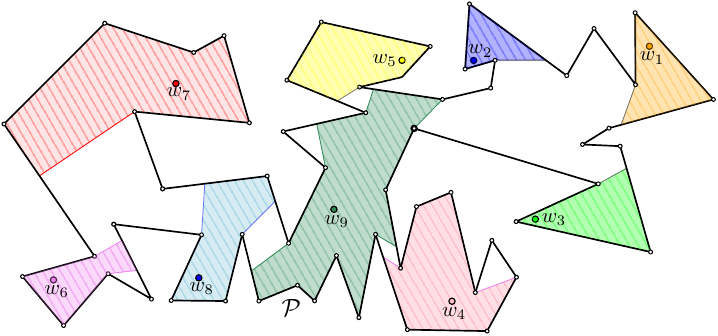}
    \caption{Example of witnesses $\{w_1,\ldots, w_9\}$  in a polygon $\mathcal{P}$.}
    \label{witfirst}
\end{figure}

\subsection{Motivation and Related Work}

Visibility-based optimization problems are a central theme in computational geometry.  
The classical \textsc{Art Gallery} problem asks for the minimum number of guards whose visibility regions cover a polygon~\cite{ghosh2007vis,ORourke87}.  
Other related problems include the \textsc{Hidden Set} problem, which seeks a maximum set of pairwise non-visible points, and the \textsc{Convex Cover} problem, which partitions a polygon into convex pieces. Let $\mathtt{ws}(\pol)$, $\mathtt{gs}(\pol)$, $\mathtt{hs}(\pol)$, and $\mathtt{cc}(\pol)$ denote the sizes of optimal witness sets, guard sets, hidden sets, and convex covers, respectively.  
These quantities satisfy
$
\mathtt{ws}(\pol)\le \mathtt{gs}(\pol)$ and $\mathtt{ws}(\pol)\le \mathtt{hs}(\pol)\le \mathtt{cc}(\pol),
$
since witnesses must be mutually exclusive in visibility and at most one hidden point can lie in any convex region.

\medskip

The \textsc{Art Gallery} problem is known to be $\mathsf{NP}$-hard~\cite{LeeL86} and APX-hard~\cite{EidenbenzSW01}.  
A major recent breakthrough by Abrahamsen, Adamaszek, and Miltzow~\cite{DBLP:journals/jacm/AbrahamsenAM22} shows that the problem is $\exists\mathbb{R}$-complete, even for polygons with integer coordinates.  
This implies that it is not in  the complexity class {\sf NP} and unlikely to admit polynomial-size discretization (certificates) or purely combinatorial algorithms unless $\mathsf{NP}=\exists\mathbb{R}$.  
Indeed, optimal guard placements may require irrational coordinates, even when only two or three guards suffice.

Despite decades of research, exact algorithms for \textsc{Art Gallery} remain limited.  
The only known approach, due to Efrat and Har-Peled~\cite{DBLP:journals/ipl/EfratH06}, runs in time $n^{\mathcal{O}(k)}$, where $k= \mathtt{gs}(\pol)$, using tools from real algebraic geometry~\cite{basu2006algorithms}, and appears difficult to replicate using purely combinatorial techniques.  Furthermore, strong inapproximability results are known: there is no $o(\log n)$-approximation for polygons with holes unless $\mathsf{P}=\mathsf{NP}$~\cite{EidenbenzSW01}. Moreover, under the Exponential Time Hypothesis ({\sf ETH}), there is no $n^{\mathcal{O}(k)}$-time algorithm for \textsc{Art Gallery}, even for simple polygons~\cite{DBLP:conf/esa/BonnetM16} and for polygons with holes~\cite{EidenbenzSW01}.

In contrast, the \textsc{Witness Set} problem imposes a fundamentally different constraint: instead of covering the polygon, it requires that visibility regions of selected points do not overlap.  
Although this definition is simple, it leads to a visibility structure that differs significantly from both guarding and hiding.  
This distinction raises the possibility that \textsc{Witness Set} admits more tractable combinatorial approaches.

\medskip 

The \textsc{Witness Set} problem was introduced by Amit, Mitchell, and Packer~\cite{AmitMitchellPacker}.  
Subsequent work has been limited.  
Daescu et al.~\cite{DBLP:journals/comgeo/DaescuFMPS19} gave a linear-time algorithm for two restricted subclasses of monotone polygons -- uni-monotone polygons and monotone mountains. 
In fact, they showed that $\mathtt{ws}(\pol)=\mathtt{gs}(\pol)$ when $\pol$ belongs to either of these classes. 
More recently, Das et al.~\cite{das2025witnesssetmonotonepolygons} extended these results to monotone polygons. 
They presented an exact algorithm for \textsc{Witness Set} running in time $r^{\mathcal{O}(k)} \cdot n^{\mathcal{O}(1)}$, based on a discretization approach, where $r$ and $n$ denote the number of reflex vertices and total vertices of the input polygon, respectively.. 
In addition, they designed a $(1+\epsilon)$-approximation scheme ({\sf PTAS}) with running time $r^{\mathcal{O}(1/\epsilon)} \cdot n^2$. 
They also showed that the \textsc{Discrete Witness Set} problem can be solved in polynomial time.

\medskip

Overall, while guarding problems are deeply connected to algebraic and continuous complexity, witness sets appear to admit a more combinatorial structure.  
This motivates a systematic study of their algorithmic and complexity-theoretic properties.

\subsection{Our Perspective}

This work is motivated by the following fundamental questions about Witness Set:

\begin{mdframed}[backgroundcolor=gray!10,
  topline=false,bottomline=false,leftline=false,rightline=false]
\centering
\begin{minipage}{0.9\linewidth}
\begin{itemize}
    \item Is \textsc{Witness Set} in $\mathsf{NP}$?
    \item Does \textsc{Witness Set} admit a finite discretization of size $n^{f(k)}$ for some computable function $f$?
\end{itemize}
\end{minipage}
\end{mdframed}

These two  questions are motivated by the stark contrast with the \textsc{Art Gallery} problem, whose decision version is $\exists\mathbb{R}$-complete~\cite{DBLP:journals/jacm/AbrahamsenAM22}, and therefore not in {\sf NP} as well as unlikely to admit polynomial-size certificates unless $\mathsf{NP}=\exists\mathbb{R}$. 

The second question concerns whether the continuous search space of candidate witness points can be reduced to a finite set. More precisely, given an instance of \textsc{Witness Set}, can we, in time $n^{f(k)}$, construct a finite point set $Q \subseteq \pol$ such that whenever $\pol$ admits a witness set of size at least $k$, the set $Q$ also contains one?

A positive answer would have strong algorithmic consequences. Combined with a polynomial-time algorithm for the discrete variant of \textsc{Witness Set}~\cite{das2025witnesssetmonotonepolygons}, such a discretization would yield an $n^{g(k)}$-time algorithm, for some computable function $g$, for deciding the existence of a witness set of size $k$.

Unlike the \textsc{Art Gallery} problem, which is governed by coverage constraints, \textsc{Witness Set} enforces pairwise non-overlap of visibility regions. This structural distinction enables a combinatorial discretization approach, which we exploit in this work.





\subsection{Our Contributions}

We present the first general algorithmic and complexity-theoretic framework
for the \textsc{Witness Set} problem in simple polygons. Our main results are summarized below.

\subparagraph{\Cref{sec:poly}. Discretization Framework.}

We show that \textsc{Witness Set} admits a finite discretization scheme. 
Our approach is based on a novel geometric ordering of witness points captured by a
\emph{neighbor witness graph}, whose chordality property plays a central role.

\begin{restatable}{theorem}{thmDisFPT}\label{thm:disfpt}
The {\sc Witness Set} problem in a simple polygon admits a finite discretization computable in
$n^{f(k)}$ time, for some computable function $f$. Moreover, within the same running time one can construct a finite set
$Q$ of size $n^{f(k)}$ such that $\pol$ contains a witness set of size $k$ if and only if $Q$ contains one.
\end{restatable}

The construction relies on a structural analysis of visibility regions, windows,
and non-visible components. Algorithm~\ref{algo_1} incrementally builds the set $Q$ by generating maximal visibility chords and their intersections, followed by midpoint insertions that ensure completeness of the discretization.

\subparagraph{\Cref{sec:witness}. Exact Algorithm.}

Using the above discretization, we obtain an exact algorithm for \textsc{Witness Set}.

\begin{restatable}{theorem}{thmfptalgorithm}\label{thm:fptalgorithm}
For simple polygons, {\sc Witness Set} can be solved in time $n^{f(k)}$, for some computable function $f$.
\end{restatable}

This follows by reducing the continuous problem to the discrete instance on $Q$. Since the discrete variant can be solved in polynomial time~\cite{das2025witnesssetmonotonepolygons}, the overall running time becomes $n^{f(k)}$. 

In particular, this implies that \textsc{Witness Set} belongs to the class {\sf XP} (slice-wise polynomial time). Recall that a parameterized problem is in {\sf XP} if it can be solved in time $n^{f(k)}$ for some computable function $f$, where $n$ denotes the input size and $k$ is the parameter. 

\subparagraph{\Cref{sec:np}. Membership in $\mathsf{NP}$.}

In contrast to the \textsc{Art Gallery} problem, which is not in $\mathsf{NP}$ unless $\mathsf{NP} = \exists\mathbb{R}$, we show that \textsc{Witness Set} lies in $\mathsf{NP}$.

\begin{restatable}{theorem}{thmNP}\label{thm:np}
The {\sc Witness Set} problem in simple polygons is in $\mathsf{NP}$.
\end{restatable}

The proof is based on a  certificate consisting of $k$ witness points, whose validity can be verified by computing visibility polygons and checking pairwise intersections in poly-time.

\subparagraph{\Cref{sec:hard}. Hardness of the Discrete Variant.}

Finally, we establish the computational hardness of the discrete version of the problem.

\begin{restatable}{theorem}{thmnph}\label{theo-hard1}
{\sc Discrete Witness Set} is $\mathsf{NP}$-complete for polygons with holes, even when restricted to rectilinear polygons.
\end{restatable}

The reduction is from \textsc{Planar Monotone Rectilinear 3SAT}~\cite{DBLP:journals/ijcga/BergK12}.

\medskip

Taken together, these results provide the first cohesive understanding of the algorithmic and complexity landscape of \textsc{Witness Set}. 
In particular, the problem admits a combinatorial discretization and exact algorithms in simple polygons, while becoming $\mathsf{NP}$-complete in the presence of holes. 
We believe that the structural insights developed here may also be useful for other visibility-based optimization problems.

\ifthenelse{\boolean{shortver}}{{\bf \em A full version of the paper containing all missing proofs and formal details  is appended at the end.}
}{}

\section{Preliminaries}
\begin{definition}[Polygon, Polygon with holes, Simple Polygon]
{\em A \emph{polygon} is a closed region in the plane bounded by a finite sequence of non-intersecting line segments, called its \emph{edges}. The endpoints of these edges are its \emph{vertices}. A \emph{polygon with holes} is a polygon whose boundary consists of more than one disjoint simple cycle: one outer boundary and one or more interior boundaries.  A polygon is \emph{simple} if it is neither self-intersecting nor contains holes.}
\end{definition}

For a polygon $\pol$, we denote its boundary by \emph{$\bd(\pol)$} and its interior by \emph{$\inte(\pol)$}. A vertex of $\pol$ is \emph{convex} if its interior angle is at most $180^\circ$; otherwise, it is \emph{reflex}. For a line segment $\overline{ab}$, we denote its interior by \emph{$iN(\overline{ab})$} $\coloneq\overline{ab}\setminus\{a,b\}$ and its boundary by \emph{$\bd(\overline{ab})$} $\coloneq \{a,b\}$.

\begin{definition}[Hairy polygon]\label{def:HairyPolygon}
{\em A \emph{hairy polygon} $\pol_h$ (\cref{hairypoly}) consists of (i) a simple polygon $\pol$ (the \emph{body}), and (ii) a finite set of line segments $\{e_1,\ldots,e_k\}$, called \emph{arms}, each attached to a distinct vertex of $\pol$ and lying outside of $\pol$ except at the attachment point.}
\end{definition}

  \begin{figure}[ht!]
        \centering
        \includegraphics[width=0.28\linewidth]{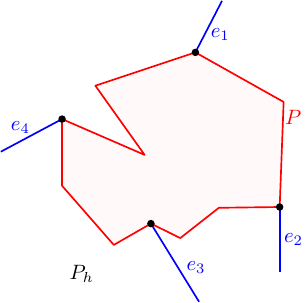}
        \caption{Example of a {hairy polygon}.}
        \label{hairypoly}
    \end{figure}
    
\begin{definition}[Visibility]
{\em Two points $x,y\in \pol$ are \emph{visible} to each other if the segment $\overline{xy}$ is contained entirely in $\pol$. The \emph{visibility region} of $p$ is $\vis(p)=\{q\in \pol : \overline{pq}\subseteq \pol\}$. For a set $S\subseteq \pol$, $\vis(S)=\bigcup_{p\in S}\vis(p)$.}
\end{definition}

If $\vis(p)$ is a hairy polygon, there may exist two distinct reflex vertices $r,r' \in \vis(p)$ such that $p, r,$ and $r'$ are collinear; see \cref{vishairypoly}.

\begin{figure}[ht!]
    \centering
    \includegraphics[width=0.7\linewidth]{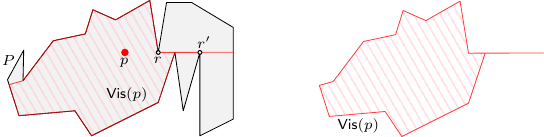}
    \caption{Example of {visibility region} which is a {hairy polygon}.}
    \label{vishairypoly}
\end{figure}


\begin{definition}[Window]
{\em A segment $l=[a,b]$ is a \emph{window} of $\vis(p)$ if $l\subseteq \bd(\vis(p))$, $l\cap \bd(\pol)=\{a,b\}$, and $l$ is maximal with this property. The endpoint closer to $p$ is called the \emph{base}, and the other is called the \emph{end}..  For an illustration, see \cref{window}.}
\end{definition}

\begin{definition}[Wall]\label{perfectprimary}
{\em A segment $l\subseteq \bd(\pol)\cap \bd(\vis(p))$ is a \emph{wall} of $\vis(p)$ if it is maximal in $\bd(\pol)\cap \bd(\vis(p))$.  For an illustration see \cref{window}.}
\end{definition}

\begin{figure}[ht!]
    \centering
    \includegraphics[width=0.5\linewidth]{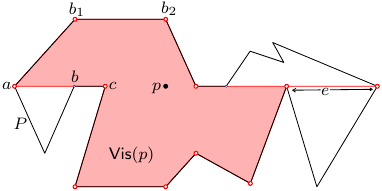}
    \caption{The red region denotes $\vis(p)$. $\overline{ab}$ represents a window of $\vis(p)$, whereas $\overline{bc}$ corresponds to a wall of $\vis(p)$. $b$ is the base of the window $\overline{ab}$ and $a$ is the end.}
    \label{window}
\end{figure}

\begin{lemma}\label{noconsecutivesecondary}
Let $e_1, e_2$ be two distinct {windows} of $\vis(p)$ such that neither of them is part of any arm (\cref{def:HairyPolygon}) of $\vis(p)$. Then $e_1$ and $e_2$ cannot be consecutive, i.e., $e_1\cap e_2=\emptyset$.

\end{lemma}

\ifthenelse{\boolean{shortver}}{}{
 \begin{proof}
    If these windows are consecutive edges of $\vis(p)$, then their common endpoint is $p$, since any window is collinear with $p$. From the definition of window, $p\in \pol $. However, this is a contradiction, since it cannot happen that $p \in \bd(\pol)$ while the segments of $\bd(\vis(p))$ incident to $p$ are not walls of $\vis(p)$.
\end{proof}  
}

\begin{mdframed}[backgroundcolor=black!10,topline=false,bottomline=false,leftline=false,rightline=false]
    Throughout Sections $3$ and $4$, unless stated otherwise, all polygons are simple.
\end{mdframed}

\section{Neighbor Witness Graph is Chordal}\label{sec:chordal}

In this section, we introduce the \emph{neighbor witness graph} associated with a witness set and establish that this graph is chordal when restricted to simple polygons. The proof relies on structural properties of non-visible components and windows.
\begin{definition}[Near polygon]\label{nearpoly}
    {\em A region $\pol'$ is called a \emph{\npo} if there is a simple polygon $\pol$ such that $\pol'\subseteq \pol$ and $\pol'=\pol\setminus B$ for some $B\subseteq\bd(\pol)$. We define $\bd(\pol')\coloneq \bd(\pol)$. Two \npo s, $\pol_1$ and $\pol_2$ are said to be \emph{attached} to each other if $iN(\pol_1)\cap iN(\pol_2)=\emptyset$ and $\bd(\pol_1)\cap\bd(\pol_2)\neq \emptyset$. Note that, the closure of a near polygon is always a simple polygon.}
\end{definition}
\begin{definition}[Witness set]
{\em In a polygon $\pol$, a finite point set $W\subseteq \pol$ is said to be a \emph{witness set} of $\pol$ if $\vis(x)\cap\vis(y)=\emptyset\text{ for all distinct } x,y\in W.$ See Figure~\ref{witt1}.}
\end{definition}
\begin{definition}[Non-visible region]
{\em For a witness set $W$ of a polygon $\pol$, we define its \emph{non-visible region} as $\nv(W)=\pol\setminus \vis(W)$. For an illustration, see \cref{witt1}. Observe that, each  component of $\nv(W)$ is a \npo.}
\end{definition}

\begin{figure}[ht!]
    \centering
    \includegraphics[width=0.7\linewidth]{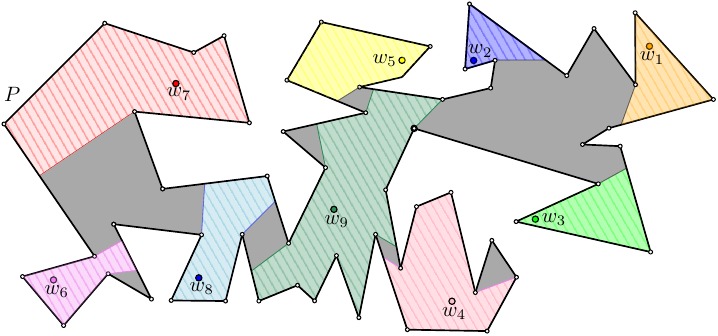}
    \caption{A polygon $\pol$ with a witness set $W=\{w_1, w_2,...,w_9\}$ and corresponding visibility region.  $\nv(W)$ is defined by the  union of the gray regions.}
    \label{witt1}
\end{figure}

 \begin{mdframed}[backgroundcolor=black!10,topline=false,bottomline=false,leftline=false,rightline=false]
     \begin{center}
         Let $W=\{w_1,\dots,w_k\}$ be a witness set of the polygon $\pol$.
     \end{center}
 \end{mdframed}

\begin{lemma}
    Let $w\in W$ be such that $\vis(w)$ has some arms. Then there exists  $w'\in \vis(w)$ such that $ (W\setminus\{w\})\cup\{w'\} $ is a witness set of size $|W|$ and $\vis(w')$ has no arms.
\end{lemma}

\ifthenelse{\boolean{shortver}}{}{
\begin{proof}
    Arms of $\vis(w)$ arise only from finitely many pairs of reflex vertices collinear with $w$. Choose $\varepsilon>0$ sufficiently small so that for all $w'\in B_\varepsilon(w)$, $\vis(w')\cap \vis(v)=\emptyset$ for all $v\in W\setminus\{w\}$. The existence of $\varepsilon$ is assured since there is always a component of $\nv(W)$ around each $e_i$, so if we move $w$ slightly in any direction, the new {visibility region} would not intersect the {visibility regions} of the other {witness points}. Pick $w'$ in the interior of $B_\varepsilon(w)$ that is not collinear with any such reflex pair. Then $\vis(w')$ has no arms, and $W'$ remains a witness set.
\end{proof}}

\begin{mdframed}[backgroundcolor=black!10,topline=false,bottomline=false,leftline=false,rightline=false]
Henceforth, w.l.o.g., we assume that for every $w\in W$, $\vis(w)$ is a simple polygon.
\end{mdframed}

We now introduce a graph, called the neighbor witness graph associated to the witness set $W$, which captures the idea of the neighborhood of a witness point.
\begin{definition}[Neighbor witness path]
{\em A \emph{neighbor witness path} between a pair of distinct witnesses $w_i$ and $w_j$ is a continuous path contained in $\vis(w_i)\cup\nv(W)\cup \vis(w_j)$.}
\end{definition}

\begin{figure}[ht!]
    \centering
    \includegraphics[width=0.7\linewidth]{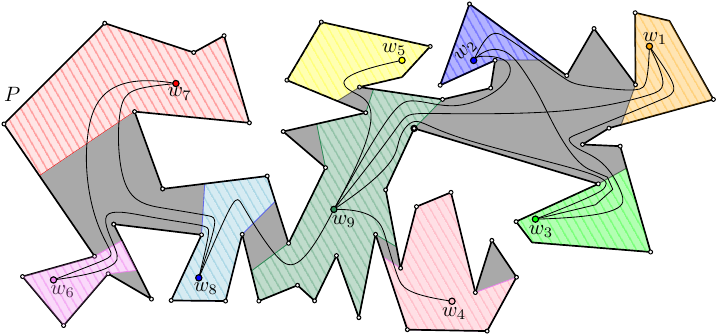}
    \caption{Black curves are {neighbor witness paths}.}
    \label{neighborwitnesspathimage}
\end{figure}

See \cref{neighborwitnesspathimage} for an illustration of neighbor witness path. Note that, it is possible that two witness points admit no neighbor witness path between them. For example, in \Cref{neighborwitnesspathimage}, there is no neighbor witness path between $w_1$ and $w_7$. The following observation follows immediately from the above definition.

\begin{observation}
    If $\gamma_{ij}$ is a neighbor witness path between $w_i$ and $w_j$, then $\gamma_{ij}\cap \vis(w_k)=\emptyset$ for each $w_k\in W\setminus\{w_i,w_j\}$.
\end{observation}

\begin{figure}[ht!]
    \centering
    \includegraphics[width=0.5\linewidth]{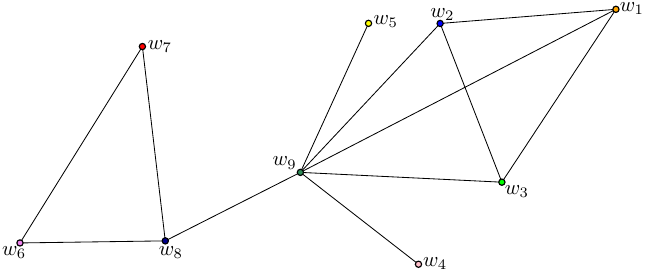}
    \caption{The graph $\gc(W)$ corresponding to the witness set in \cref{witt1}.}
    \label{graph1}
\end{figure}
\begin{definition}[Neighbor witness graph]\label{neighborwitnessgraph}
{\em Let $W$ be a witness set. The \emph{neighbor witness graph} $\gc(W)$ is the graph with vertex set $W$, where distinct vertices $w_i,w_j$ are adjacent if and only if a neighbor witness path exists between them. See \cref{graph1}.}
\end{definition}

\begin{definition}[Neighbor witness]\label{local}
{\em For $w\in W$, a point $v\in W\setminus\{w\}$ is a \emph{neighbor witness} of $w$ if $vw$ is an edge of $\gc(W)$. The set of all such neighbors is denoted by $\vil(w)$.}
\end{definition}

Observe that if a component of $\nv(W)$ attaches to $\vis(w)$ for some 
$w \in W$, then this attachment can occur only along a {window} 
of $\vis(w)$. This is because no portion of $\bd(\pol)$ can serve as a common 
boundary between a visibility region and a component of $\nv(W)$. 
Furthermore, whenever a window participates in such an attachment, 
the connection necessarily extends over its entire length.

\begin{lemma}[Uniqueness of attached component]\label{uniquenessofsecondary}
Let $w\in W$, and let $NV_1,NV_2$ be connected components of $\nv(W)$ attached to $\vis(w)$. Let $e_1,e_2$ be the windows of $\vis(w)$ along which $NV_1,NV_2$ are attached. Then $NV_1 = NV_2$ if and only if $e_1=e_2.$
\end{lemma}

\ifthenelse{\boolean{shortver}}{}{
\begin{proof}
    The forward implication is immediate: whenever a component is attached to a visibility region, it is necessarily tied to exactly one window. This is a direct consequence of the fact that $\pol$ is a simple polygon and, by \cref{noconsecutivesecondary}, two windows cannot occur consecutively. For the reverse direction, assume that $e_1=e_2=e$. Observe that both $NV_1, NV_2$ are attached to $\vis(w)$ on the same side of $e$, otherwise, if $NV_1$ and $NV_2$ are on either side of $e$ then one of them must have a non-empty intersection with $\vis(w)$. Now $NV_1, NV_2$ lie on the same side of $e$ and the entire length of $e$ is involved for both attachments, so it follows that $NV_1\cap NV_2\neq \emptyset$. Since distinct connected components are necessarily disjoint, we conclude that $NV_1 = NV_2$.
\end{proof}}

\begin{lemma}\label{finiteintersectionwithbdry}
Let $a,b\in \pol$, and let $\gamma_1,\gamma_2\subseteq \pol$ be two simple paths from $a$ to $b$. Let $R$ be the closed region bounded by $\gamma_1\cup\gamma_2$. Then the interior of $R$ contains no points of $\bd(\pol)$.
\end{lemma}

\ifthenelse{\boolean{shortver}}{}{\begin{proof}
If $\bd(\pol)$ contained a point strictly inside $R$, then a portion of $\bd(\pol)$ would lie entirely inside $\pol$, contradicting the fact that $\pol$ is a simple polygon.
\end{proof}}

\begin{lemma}[Unique component between neighbors]\label{UniquenessComponentBetweenTwo}
Let $v,w\in W$ be neighbor witnesses. Then there exists exactly one component of $\nv(W)$ attached to both $\vis(v)$ and $\vis(w)$.
\end{lemma}

\ifthenelse{\boolean{shortver}}{}{
\begin{proof}
    Since $ v $ and $ w $ are {neighbor witnesses} of each other, there must exist a component of $ \nv(W) $ that lies between their respective {visibility regions}. Our goal is to establish the uniqueness of such a component.    
    \begin{figure}[ht!]
    \centering
    \includegraphics[width=0.5\linewidth]{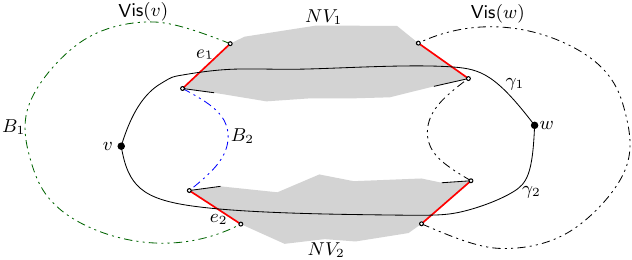}
    \caption{$\bd(\vis(v))\setminus\{e_1,e_2\}=B_1\cup B_2$.}
    \label{obsv4.10}
    \end{figure}
    Assume, for contradiction, that there are two distinct components $ NV_1 $ and $ NV_2 $ of $ \nv(W) $ that are attached to both $ \vis(v) $ and $ \vis(w) $. Now, using \Cref{uniquenessofsecondary}, suppose that $ NV_1 $ is attached to $ \vis(v) $ via the {window} $ e_1 $. Similarly, let $ NV_2 $ be attached to $ \vis(v) $ via the {window} $ e_2 $. Let $ \gamma_1 $ be a {neighbor witness path} between $ v $ and $ w $ that passes through $ e_1 $ and $ NV_1 $, and let $ \gamma_2 $ be another {neighbor witness path} that passes through $ e_2 $ and $ NV_2 $. Since $ e_1 $ and $ e_2 $ are both {windows} of $ \vis(v) $, they are not consecutive by \Cref{noconsecutivesecondary}. Therefore, $ \bd(\vis(v)) \smallsetminus \{e_1, e_2\} $ is the union of two disjoint {simple polygonal chains}, denoted $ B_1 $ and $ B_2 $, each containing a non-empty portion of $ \bd(\pol) $. Now, the interior of the region enclosed between $ \gamma_1 $ and $ \gamma_2 $ must contain exactly one of $B_1$ or $B_2$, which contradicts \Cref{finiteintersectionwithbdry}. See \cref{obsv4.10} for an illustration.
\end{proof}}

\begin{lemma}[Unique component for cycles]\label{uniquecomponentinsidecycle}
Let $C=(w_1,\dots,w_m)$ be a cycle of length $m\ge 3$ in $\gc(W)$. Then there exists exactly one component $NV_C$ of $\nv(W)$ such that $NV_C$ is attached to $\vis(w_i)$ for every $i$.
\end{lemma}

\begin{figure}[ht!]
        \centering
        \includegraphics[width=0.5\linewidth]{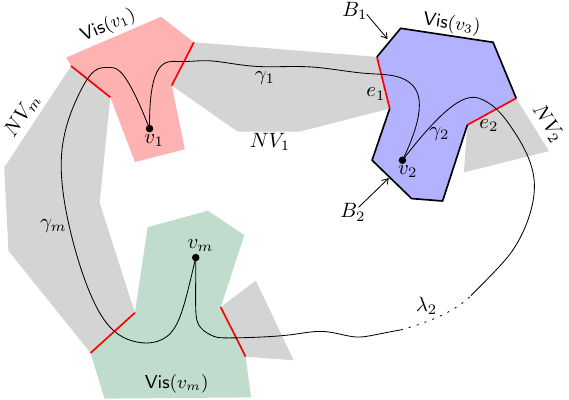}
        \caption{$B_2$ is contained inside the area enclosed by $\lambda_1=\gamma_1\cup \gamma_2$ and $\lambda_2 =\bigcup_{j=2}^{m-1}\gamma_j$. }
        \label{perfectprimarycontained}
\end{figure}

\begin{proof}
    Let $C$ be a cycle of length $m$ with vertices $v_1, \ldots, v_m \in W$. For each $i$, let $NV_i$ denote the component of $\nv(W)$ that is attached to both $\vis(v_i)$ and $\vis(v_{i+1})$ (indices taken modulo $m$). If these components were not all identical, then there would exist some $i$ for which $NV_i \ne NV_{i+1}$. Without loss of generality, assume $i = 1$; that is, $NV_1 \neq NV_2$. Let $\gamma_i$ be a {neighbor witness path} (corresponding to $NV_i$) between $v_i$ and $v_{i+1}$.  Let $NV_1$ and $NV_2$ be attached to $\vis(v_2)$ on the {windows}, $e_1$ and $e_2$ respectively. From \Cref{uniquenessofsecondary}, we have $NV_1\neq NV_2$ which implies $e_1\neq e_2$.  Now we consider two paths between $v_2$ and $v_m$. Let the first path $\lambda_1$ be $\gamma_1 \cup \gamma_m$, and let the second path $\lambda_2$ be $\bigcup_{j=2}^{m-1} \gamma_j$; see \Cref{perfectprimarycontained} for an illustration. Since $e_1$ and $e_2$ are both {windows} of $\vis(v_2)$, they cannot be consecutive by \Cref{noconsecutivesecondary}. Therefore, $\bd(\vis(v_2)) \smallsetminus \{e_1, e_2\}$ is the union of two disjoint {simple polygonal chains}, denoted $B_1$ and $B_2$ (see \Cref{perfectprimarycontained}), each containing at least one wall of $\vis(v_2)$. The interior of the region enclosed between $\lambda_1$ and $\lambda_2$ must contain exactly one of the $B_i$, $i \in \{1,2\}$, which contradicts \Cref{finiteintersectionwithbdry}. Hence $NV_1 = \cdots = NV_m$; that is, $\vis(v_1), \ldots, \vis(v_m)$ are attached to the same component.
\end{proof}
\begin{definition}[Chordal Graph]
    {\em A graph $G$ is said to be \emph{chordal} if every induced cycle in the graph  have exactly three vertices.}
\end{definition}
\begin{theorem}\label{chordalll}
For any witness set $W\subseteq \pol$, the neighbor witness graph $\gc(W)$ is chordal.
\end{theorem}

\begin{proof}
Consider any cycle $C$ of length $m\ge 4$ in $\gc(W)$. By Lemma~\ref{uniquecomponentinsidecycle}, there exists a unique component $NV_C$ of $\nv(W)$ attached to all visibility regions of vertices of the cycle. For any two distinct vertices $w_i,w_j$ in the cycle, both attach to $NV_C$. Thus, a neighbor witness path exists between them, implying the edge $w_iw_j$ belongs to $\gc(W)$. Every cycle of length at least four therefore contains a chord, so $\gc(W)$ is chordal.
\end{proof}

\section{{\sc Witness Set} in  Simple Polygons} \label{sec:poly}

 In this section, our aim is to find an algorithm for the {\sc  Witness Set} problem in polygons. To do that, we find a discrete set within a polygon $\pol$ that suffices to contain a solution of the {\sc Witness Set} in $\pol$. In \cref{sec:discrete}, we show that any witness point corresponding to a simplicial vertex in the neighbor witness graph can be replaced by a vertex of the polygon, without affecting the witness property of the set. The structural argument relies on the uniqueness of windows and non-visible components established in the previous section. In \cref{sec : algo}, we work with an arbitrary simple polygon $\pol$ and construct a finite point set $Q \subseteq \pol$ with the following property: For every witness set $W \subseteq \pol$ of size at most $k$, there exists a witness set $W' \subseteq Q$ with $|W'| = |W|$. Finally, in \cref{sec:witness}, we solve the {\sc Witness Set} problem for simple polygons.

\subsection{Replacing a Simplicial Witness by a Polygon Vertex} \label{sec:discrete}
\textbf{Notation:} For $w\in W$ and $v\in \vil(w)$, \emph{$e_{w,v}$} denotes the window of $\vis(w)$ along which the non-visible component between $\vis(w)$ and $\vis(v)$ is attached with $\vis(w)$.
\begin{lemma}\label{intersectper}
    Let $w\in W$ and $v\in \vil(w)$. Then every neighbor witness path between $w$ and $v$ must intersect both $e_{w,v}$ and $e_{v,w}$.
\end{lemma}

\ifthenelse{\boolean{shortver}}{}{
\begin{proof}
    Let $\gamma$ be a neighbor witness path between $w$ and $v$. Since $\gamma$ lies entirely in $\vis(w)\cup\nv(W)\cup\vis(v)$, there exists a window $e$ of $\vis(w)$ through which $\gamma$ leaves $\vis(w)$ and enters the component $NV$ of $\nv(W)$ between $\vis(w)$ and $\vis(v)$. Then $NV$ is attached to $e$, and hence $e=e_{w,v}$ using \cref{uniquenessofsecondary}. Thus, $\gamma$ intersects $e_{w,v}$. Similar argument applies to $e_{v,w}$.
\end{proof}}

\begin{lemma}\label{something}
    Let $v', w \in W$. If a point $x \in \vis(w)$ is visible from some $y \in \vis(v')$, then there exists $v \in \vil(w)$ such that $x \in \vis(e_{v,w})$.
\end{lemma}

\ifthenelse{\boolean{shortver}}{}{\begin{proof}
We have two following cases.
\begin{description}
    \item[Case I : $v'\in \vil(w)$.]   In this case, $\overline{wx} \cup \overline{xy} \cup \overline{yv'}$ defines a neighbor witness path between $w$ and $v'$. By \cref{intersectper}, it intersects $e_{v',w}$. Let $z'$ be the intersection point of $\overline{xy}$ and $e_{v',w}$. Hence, $x$ is visible from $z' \in e_{v',w}$. See \cref{impobsv}.
    
\begin{figure}[ht!]
    \centering
    \includegraphics[width=1\linewidth]{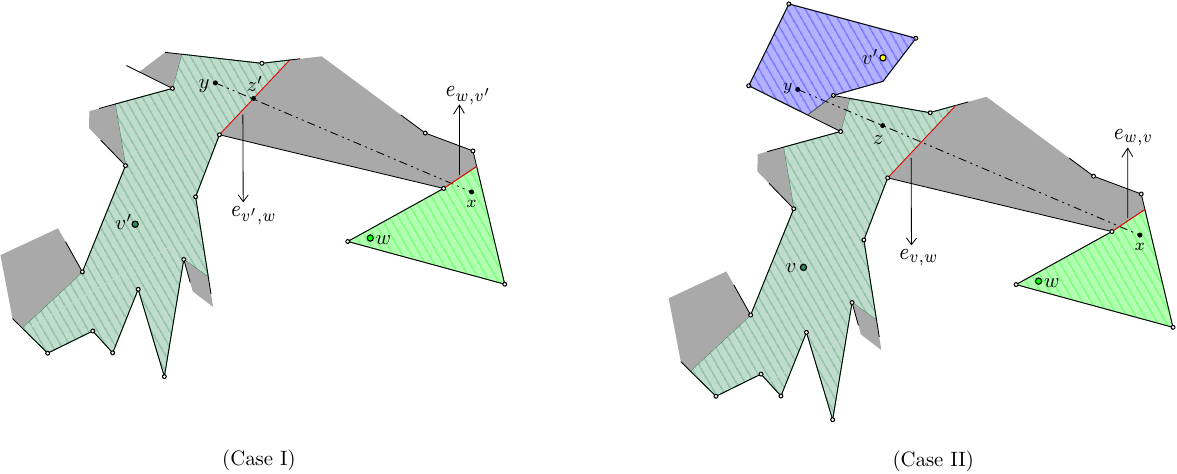}
    \caption{Proof of \cref{something}.}
    \label{impobsv}
\end{figure}

\item[Case II : $v'\notin \vil(w)$] If $\,\overline{xy}\,$ intersects no $\vis(v)$ for $v\in \vil(w)$, then the path $\gamma=\overline{wx}\cup \overline{xy}\cup \overline{yv'}$ lies entirely in $\vis(w)\cup \nv(W)\cup \vis(v')$, forming a neighbor witness path between $w$ and $v'$, contradicting $v'\notin \vil(w)$. Hence, there exists $v \in \vil(w)$. Let $z\in \vis(v)\cap \overline{xy}$. The path $\overline{wx}\cup \overline{xz}\cup \overline{zv}$ is a neighbor witness path, so it passes through the window $e_{v,w}$ using \cref{intersectper}. Thus, $\overline{xz}$ intersects $e_{v,w}$, which implies $x$ is visible from $e_v$. See \cref{impobsv}.   
\end{description}
\end{proof}
}

\begin{lemma}\label{WitnessToEachOther}
    Let $v, w $ be two points in the polygon $\pol$. If $v $ and $w$ are not visible to each other, and $w$ is not visible from any window of $\vis(v)$, iff $\vis(v) \cap \vis(w) = \emptyset$.
\end{lemma}
\ifthenelse{\boolean{shortver}}{}{\begin{proof}
   Assume for contradiction, that there exists a point $x\in \pol$ such that it is in $\vis(v)\cap\vis(w)$. As $x\in \vis(v)$ and the line segment $\overline{xw}$ completely lies inside $\pol$, $\overline{xw}$ must intersect some {window}, $e$ of $\vis(v)$ at a point $y$ (say). So $w$ is visible from $y\in e$ which is a contradiction to the assumption.

   For the reverse direction, since $\vis(v) \cap \vis(w) = \emptyset$, it follows immediately that $v$ and $w$ are not visible to each other. Furthermore, if $w$ were visible from a point $x$ lying on a window of $\vis(v)$, then $x\in\vis(v)\cap\vis(w)$, which contradicts our assumption.
\end{proof}
}

\begin{corollary}\label{proofofcontainment}
    Let $w\in W$ be a witness point and $w'\in \vis(w)\setminus\bigcup_{v\in \vil(w)}\vis(e_{v,w})$. Then $(W\setminus\{w\})\cup\{w'\}$ is a witenss set of size $|W|$.
\end{corollary}

\ifthenelse{\boolean{shortver}}{}{
\begin{proof}
    Let $v'\in W\setminus \{w\}$. If $v'$ and $w'$ are visible to each other then $w'\in \vis(v')\cap\vis(w)$ which contradicts the fact that $v',w\in W$. Now if $w'$ is visible from some window $e$ of $\vis(v')$ then there exists a point $y\in e~(\subset \vis(v'))$ such that $w'$ is visible from $y$. Using \cref{something}, there exists $v\in \vil(w)$ such that $x\in \vis(e_{v,w})$ which is a contradiction to our assumption. So, for each $v'\in W\setminus\{w\}$, $\vis(v')\cap \vis(w')=\emptyset$ using \cref{WitnessToEachOther}. Hence $(W\setminus\{w\})\cup\{w'\}$ is a witenss set of size $|W|$.
\end{proof}}

\begin{definition}[Restricted visibility]
{\em For a point $p$ and a segment $L\subseteq \pol$, we define $\vis(L;p)\coloneq\vis(L)\cap \vis(p).$  For an illustration, see \cref{restvis}.}
\end{definition}

\begin{figure}[ht!]
    \centering
    \includegraphics[width=0.3\linewidth]{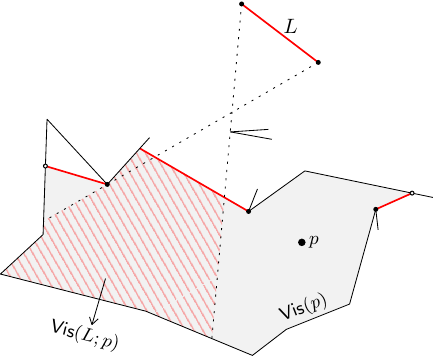}
    \caption{The gray region captures $\vis(w)$ and the red region is $\vis(L;p)$.}
    \label{restvis}
\end{figure}
Let $w\in W$ and $e$ be a window of $\vis(w)$ with base $d$, where $d'$ is the intersection of the extension of $e$ with $\partial\vis(w)$. Let the extension of $e$ divide $\vis(w)$ into two simple polygons $V_1^w(e)$ and $V_2^w(e)$. One of them contains $e$ on its boundary. Without loss of generality, let $e\subset \bd(V_1^w(e))$. See \cref{pokpok}. Under this setup, we have the following lemma :
\begin{lemma}\label{onlyoneside}
Let $v_1,\ldots,v_m\in \vil(w)$ be such that $e_{w,v_i}=e$ for each $i\in [m]$. Then $\vis(e_{v_i,w};w)\subseteq V_1^w(e)$ for all $i \in [m]$.
Moreover,
$
\overline{dd'}\cap \bigcup_{i=1}^m \vis(e_{v_i,w};w)\in \{\emptyset,\{d\}\}$.
\end{lemma}

\begin{figure}[ht!]
    \centering
    \includegraphics[width=0.5\linewidth]{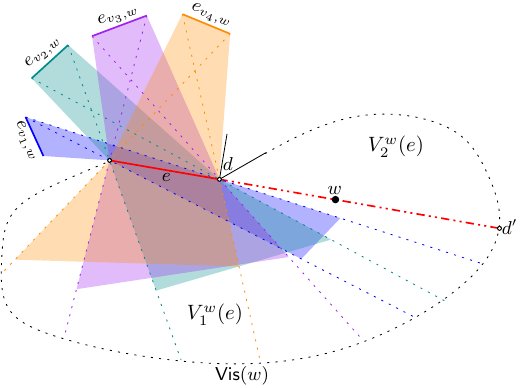}
    \caption{For each $i$, $\vis(e_i;w)\subset V_1$. }
    \label{pokpok}
\end{figure}
    
\ifthenelse{\boolean{shortver}}{}{
\begin{proof}
Let a point $y\in \vis(w)$ is visible from some point $x\in e_{v_i,w}$. Then $y\in V_1^w(e)$, otherwise, if $y\in V_2^w(e)\setminus \overline{dd'}$ then $\overline{xy}$ cannot remain in $\pol$ because of the fact that $V_2^w(e)$ and the edges incident to $d$ lie on the same side of the extension of $e$. Thus, $\vis(e_{v_i,w})\subseteq V_1^w(e)$.

If $z\in \overline{dd'}\setminus\{d\}$ were visible from some $x\in e_i$, 
then $w$ would be visible from $z$, since $z$ lies on a boundary ray from $w$.
But then $z\in \vis(w)\cap \vis(v_i)$, contradicting $w,v_i\in W$.
Hence, the intersection condition holds.
\end{proof}}

\begin{definition}[Simplicial Vertex]
    {\em In a graph $G$, a vertex $v$ is said to be a \emph{simplicial vertex} if the subgraph of $G$ induced by the neighbors of $v$ is a complete graph.}
\end{definition}
\begin{lemma}\label{onlyoneperfectwindowinvolvedforsimplicialvertex}
Let $w\in W$ be a simplicial vertex of $\gc(W)$.
Then there exists a window $e$ of $\vis(w)$ such that $e_{w,v}=e$ for every $v\in \vil(w)$.
\end{lemma}

{\begin{proof}
Since $w$ is simplicial, the induced subgraph on $\{w\}\cup \vil(w)$ is a clique.
Thus, any two vertices of this subgraph lie in a common cycle.
By \cref{uniquecomponentinsidecycle}, 
all such vertices share a single non-visible component, hence a single attachment window using \cref{uniquenessofsecondary}.
\end{proof}

\begin{lemma}\label{ReplacingSimplicialVertexByAVertex}
Let $w\in W$ be a simplicial vertex in $\gc(W)$.
Then there exists a polygon vertex $w^*\in \vis(w)$ such that
$
W' := (W\setminus\{w\})\cup \{w^*\}
$
is a witness set of $\pol$.
\end{lemma}

\begin{proof}
Let $e$ be the window of $\vis(w)$ such that $e_{w,v}=e$ for each $v\in \vil(w)$ guaranteed by \cref{onlyoneperfectwindowinvolvedforsimplicialvertex}. Let $d$ be the base of $e$ and $d'$ be the intersection point of the extension of $e$ with $\partial\vis(w)$. Using \cref{onlyoneside}, $\vis(e_{v,w};w)\subseteq V_1^w(e)$ for each $v\in \vil(w)$ and hence $(V_2^w(e)\setminus\{d\})\subseteq \vis(w)\setminus\bigcup_{v\in \vil(w)}\vis(e_{v,w})$. Now we show that $V_2^w(e)\setminus\{d\}$ contains a polygon vertex. 
\begin{figure}[ht!]
        \centering
        \includegraphics[width=0.5\linewidth]{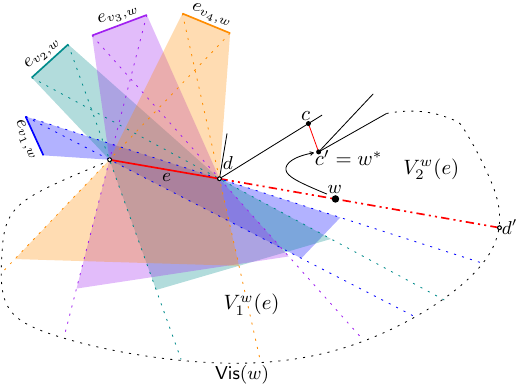}
        \caption{Proof of \cref{ReplacingSimplicialVertexByAVertex}.}
        \label{simplicialtovertex}
\end{figure}
Let $c$ be the vertex of $V_2^w(e)$ adjacent to $d$ (distinct from $d'$). If $c$ is a polygon vertex, choose $w^*=c$. Otherwise $c$ is an endpoint (not the base, since $c$ is not a polygon vertex) of a window $e'$ of $\vis(w)$. Let $c'$ be the base of $e'$. Since two windows cannot be consecutive, $c'\neq d$. Now $c$ and $e'$ lies on the same side of the extended $e$, since $e'$ does not cross the extended $e$. So $c'\in V_2^w(e)\setminus \{d\}$. We take $w^*=c'$. Thus, using \cref{proofofcontainment}, $W'$ is a witness set and $w^*$ is a polygon vertex since $c'$ is a reflex vertex. See \cref{simplicialtovertex}.
\end{proof}
\subsection{Placing Each Witness to an Algorithmically Generated Point}\label{sec : algo}
\subsubsection{The Discretization Algorithm}

The set $Q$ serves as a discretization of all possible witness locations.
The construction is entirely combinatorial and depends only on the input polygon
$\pol$ and given integer parameter $k$, the size of the witness set we plan to
replace. The following algorithm generates the set $Q$.

\IncMargin{1em}
\begin{algorithm}[ht!]\label{witgen}
\SetKwInOut{Input}{Input}\SetKwInOut{Output}{Output}

\Input{A simple polygon $\pol$ and an integer $k\in \mathbb{N}$.}
\Output{A finite point set $Q$.}

$E(\pol) \coloneqq$ set of all edges of $\pol$; \\
$R \coloneqq$ set of all reflex vertices of $\pol$; \\[2mm]

Initialization: $Q \gets V(\pol)$; $L \gets E(\pol)$; \;

\For{$i \gets 1$ \KwTo $2(k-1)$}{
    $L' \gets$ set of all maximal chords joining $v\in Q$ to $r\in R$ if $r$ is visible from $v$;\\
    $Q' \gets \{\, l'\cap l \mid l' \in L',~ l \in L,~ l'\cap l \neq \emptyset \,\}$;\\
    $L \gets L \cup L'$;\\
    $Q \gets Q \cup Q'$;\\
}

Choose any visible pair $v_1,v_2\in Q$;\\
$M \gets \left\{\, \frac{v_1+v_2}{2} \,\right\}$; \\
$Q \gets Q \cup M$;\\

\For{$i \gets 1$ \KwTo $2(k-2)$}{
    $L'' \gets$ set of all maximal chords joining $v\in Q$ to $r\in R$ if $r$ is visible from $v$;\\
    $Q'' \gets \{\, l''\cap l \mid l'' \in L'',~ l\in L,~ l''\cap l \neq \emptyset \,\}$;\\
    $L \gets L \cup L''$;\\
    $Q \gets Q \cup Q''$;\\

    Choose any visible pair $v_1,v_2\in Q$;\\
    $M \gets \left\{\, \frac{v_1+v_2}{2} \,\right\}$;\\
    $Q \gets Q \cup M$;\\
}

\Return{$Q$};

\caption{\texttt{WitGen}$(\pol,k)$: Generation of the discretizing set $Q$.}
\label{algo_1}
\end{algorithm}
\DecMargin{1em}

\subparagraph{Brief Description of the Algorithm :}

Here we provide a short description of \Cref{algo_1}. Initially, we set $Q = V(\pol)$ and $L = E(\pol)$. Iteratively, we construct chords from every point in the current set $Q$ (as updated in the previous iteration) by joining it with the reflex vertices visible from that point and extending the resulting segment until it meets the boundary of $\pol$. We refer to this procedure of generating chords from points in $Q$ as forming the \emph{arrangement}. In each iteration of the first \textbf{for} loop ($2(k-1)$ times), we draw an arrangement and then update $Q$ by adding the intersection points between the newly created chords and the existing line segments of $L$. These points serve as candidates for replacing the first $k-1$ witness points (by clone points, see \cref{def_clone}) among the $k$ witness points. A clone of a point refers to another point that lies infinitesimally close to the original one. Observe that the lines produced by the arrangement form a superset of the boundaries of the visibility regions of the points in $Q$ from the previous iteration. Next, between the two \textbf{for} loops, we compute the midpoints of all pairs of mutually visible points and update $Q$ by inserting these midpoints. This step replaces the last witness point with an algorithmically generated point (rather than a clone). In each iteration of the second \textbf{for} loop ($2(k-2)$ times), we construct arrangements, update $Q$, compute the midpoints of all mutually visible pairs of points in $Q$, and insert these midpoints into $Q$. This procedure adjusts the previous positions of the witness points to appropriate points in $Q$ (that is, from clones to mid-points).

The following observation is immediate from the construction of $Q$:
\begin{observation}\label{inq}
The following  holds for $Q$.
 Let $v_1, v_2 \in Q$, where $Q$ is obtained from an intermediate iteration. Then, in the next iteration, $Q$ includes the vertices of $\vis(v_i)$ for each $i$, as well as the intersection points of the edges of $\vis(v_1)$ and $\vis(v_2)$. In addition, in every iteration of the second for loop,  we add midpoints of all mutually visible pairs of points to $Q$.

    \end{observation}
\subsubsection{Correctness of the Construction}

To prove the correctness of the algorithm, we consider an arbitrary witness set $W$ and show that there exists a corresponding replacement of $W$ within $Q$. 

\begin{definition}[Perfect Elimination Ordering]
    {\em A \emph{perfect elimination ordering} of a graph $G$ on $n$ vertices is an ordering $\{v_1, \ldots, v_n\}$ of $G$'s vertices such that, for each $i$, the vertex $v_i$ is simplicial in the subgraph induced by $\{v_i, \ldots, v_n\}$. It is well known that every chordal graph has a perfect elimination ordering, and can be computed in polynomial time. Consequently, every chordal graph contains a simplicial vertex.}
\end{definition}

Let $|W|=k$ and $G_1 \coloneqq \gc(W)$. Since $G_1$ is chordal (from \cref{chordalll}), let $\{w_1,\ldots, w_k\}$ be a perfect elimination ordering of $G_1$. Define $G_i$ to be the subgraph of $G_1$ induced by the vertex set $W_i=\{w_i,\ldots, w_k\}$. Clearly, $w_i$ is a simplicial vertex of $G_i$. 

Now consider the analogous argument for the polygon $\pol$ with the witness set $W=\{w_1,\dots,w_k\}$, similar to the construction of the sequence of graphs $G_i$. Initially, we begin with the witness set $W=W_1=\{w_1,\dots,w_k\}$ and its associated non-visibility region $\nv(W_1)=\pol\setminus \left(\bigcup_{j=1}^k \vis(w_j)\right)$. After removing $w_1$, we update the witness set to $W_2$, and the corresponding non-visible region becomes $\nv(W_2)\coloneq \nv(W_1)\cup \vis(w_1) .$ More generally, at the $(i-1)$th step, after removing $w_{i-1}$, the witness set $W_{i-1}$ is updated to $W_i$ and the corresponding non-visible region is updated as $\nv(W_i)\coloneq \nv(W_{i-1})\cup \vis(w_{i-1}).$

From the lemma below, we obtain $G_i = \gc(W_i)$.
\ifthenelse{\boolean{shortver}}{}{Although this may seem immediate, one might worry that deleting some of the witness points $w_1,\dots,w_{i-1}$ could create new edges between the remaining witnesses; however, this does not occur in our setting.}
Thus, deleting a witness point from $\gc(W_i)$ is equivalent to removing the corresponding node from $G_i$. For each $w \in W_i$, we define
\emph{$\vil_i(w)$} as the neighbors of $w$ in $G_i$.

\begin{lemma}\label{equivalbetw}
    $G_i=\gc(W_i)$ for each $i \in [k]$.
\end{lemma}

\ifthenelse{\boolean{shortver}}{}{\begin{proof}
    We prove this by induction. For the base case we have $G_1=\gc(W_1)$ which follows from the definition of $G_1$ and $\gc(W_1)$.
    
    Let for some $l~(1\le l<k)$, $G_l=\gc(W_l)$. We prove $G_{l+1}=\gc(W_{l+1})$. We have that the vertex set of $\gc(W_{l+1})$ is the same as the vertex set of $G_{l+1}$. Now an edge which is in $G_{l+1}$ must be also in $\gc(W_{l+1})$ since for $j_1,j_2>l$, if there is a path $\gamma\subset \vis(w_{j_1})\cup\nv(W)\cup\vis(w_{j_2})$ (which produces an edge in $G_1$ as well as in $G_{l+1}$) between $w_{j_1}$ and $w_{j_2}$ then $\gamma$ is also contained inside $\vis(w_{j_1})\cup\nv(W_{l+1})\cup\vis(w_{j_2})$ as $\nv(W)\subset \nv(W_{l+1})$.

    Now assume, for contradiction, removal of $w_l$ and addition of $\vis(w_l)$ with $\nv(W_l)$, gives rise to a new edge between $w_{j_1}$ and $w_{j_2}$ which was not in $G_l$ (and hence not in $G_{l+1}$ either). Then $w_{j_1}$ and $w_{j_2}$ must be in $\vil_l(w_l)$. But since $w_l$ is simplicial in $G_l$, so $\vil_l(w_l)$ forms a clique in $G_l$. Thus no new edge can be added between $w_{j_1}$ and $w_{j_2}$ which is a contradiction to our assumption. And hence $G_{l+1}=\gc(W_{l+1})$. 
\end{proof}}


\subparagraph{Step 1 (Replacing vertices by clones):}

\subparagraph{\underline{Base Case}:}
Invoking \cref{ReplacingSimplicialVertexByAVertex}, we may assume, without loss of generality, that $w_1 \in V(\pol)$. Thus we define $w_1^*\coloneq w_1\in Q$, and we have that $\{w_1^*, w_2, \ldots,w_k\}$ is a witness set.
\begin{definition}[Clone]\label{def_clone}
 		{\em For a point $v \in \pol$, we denote by $Cl(v) \in \pol$ a point that is infinitesimally close to $v$. We refer to $Cl(v)$ as a \emph{clone} of $v$.}  
 \end{definition} 

For a point $v$, since $Cl(v)$ is infinitesimally close to $v$, we may regard $Cl(v)$ as essentially coincident with $v$ for subsequent arguments.

\subparagraph{\underline{Inductive Step:}} For this step, assume that $\{w_1^*,\ldots, w_i^*,w_{i+1},\ldots, w_k\}$ is a witness set of size $k$, such that for every $1 \leq j \leq i$, $w_j^*$ is a clone of some point $v_j\in Q$ and $w_j^*\in \vis(w_j)$. We show that there exists $v_{i+1}\in Q$ and a clone $w_{i+1}^*$ of $v_{i+1}$ such that $\{w_1^*,\ldots, w_{i+1}^*,\ldots, w_k\}$ is a witness set of size $k$ and $w_{i+1}^*\in \vis(w_{i+1})$.



We want to find a point of $Q$ from $M_{i+1}\coloneq \vis(w_{i+1})\setminus\bigcup_{v\in N(w_{i+1})}\vis(e_{v,w_{i+1}})$ for the replacement of $w_{i+1}$, as justified by \cref{proofofcontainment}. Now $\vis(w_{i+1})\setminus\bigcup_{{w_j^*\in \vil(w_{i+1}),\,j\leq i }}\vis(e_{w_j^*,w_{i+1}})$ is a disjoint union of some \npo s $P_1, \ldots, P_{n_{i+1}}$. Furthermore, $w_{i+1}$ is contained in exactly one of these \npo s. Without loss of generality, assume $w_{i+1} \in P_1$. We define the closure of the \npo, $P_1$ to be the \emph{$i$-th movable region} associated with $w_{i+1}$, denoted by $MV_i(w_{i+1})$ (see \cref{recursivestepa}). 

\begin{figure}[ht!]
  \centering
  \subfloat[Example of $MV_i(w_{i+1})$. $j_1,j_2\le i$]{%
    \includegraphics[width=0.45\linewidth]{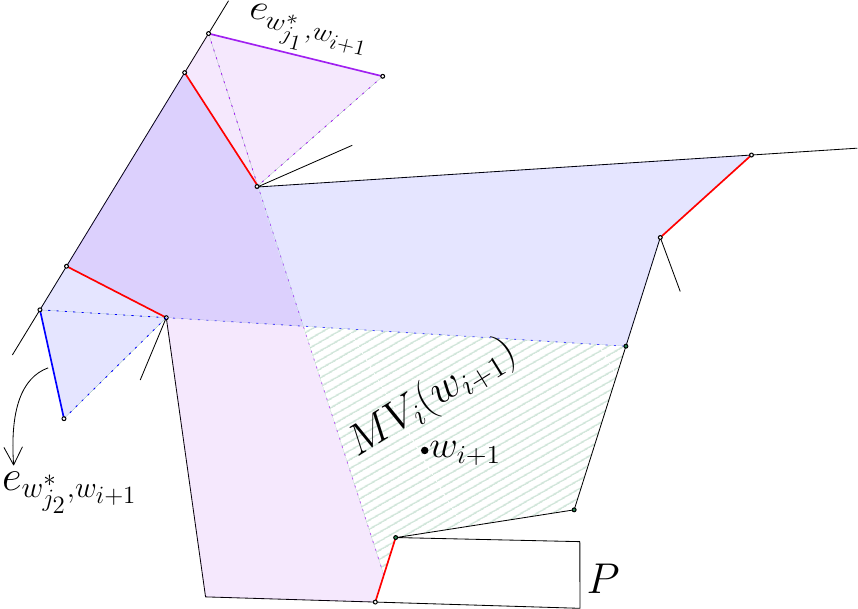}
    \label{recursivestepa}
  }
  \hfill
  \subfloat[Example of $MV_i^{(1)}$ and $MV_i^{(2)}$.]{%
    \includegraphics[width=0.45\linewidth]{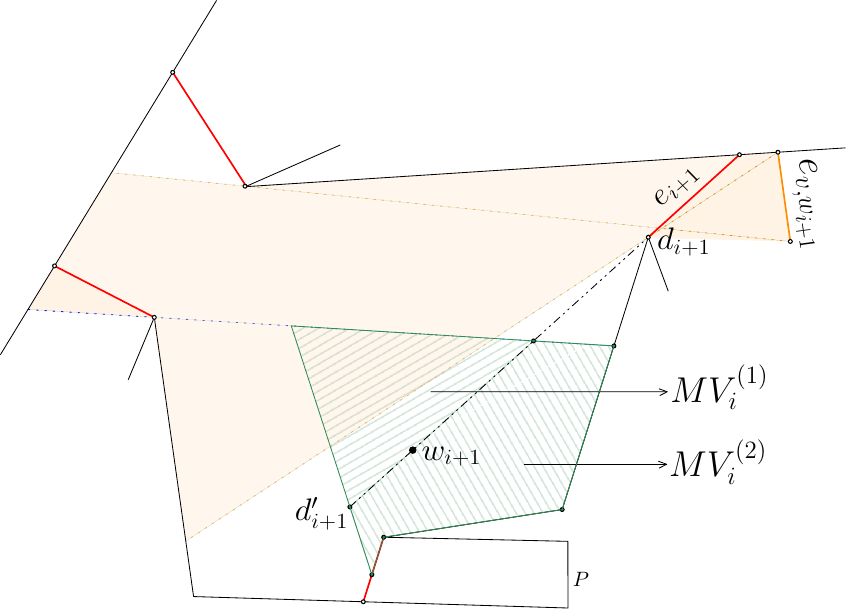}
    \label{recursivestep}
  }
  \caption{Illustration of the inductive step.}
\end{figure}


    

Since $w_{i+1}$ is a simplicial vertex of $G_{i+1}$, \cref{equivalbetw} and \cref{onlyoneperfectwindowinvolvedforsimplicialvertex} ensures the existence of the window $e_{i+1}$ of $\vis(w_{i+1})$ such that for every $v \in \vil_{i+1}(w_{i+1})$ in $G_{i+1}$, $e_{w_{i+1},v}=e_{i+1}$. We show that the region $MV_i(w_{i+1})\setminus\bigcup_{v\in\vil_{i+1}(w_{i+1})}\vis(e_{v,w_{i+1}})\,(\subseteq M_{i+1})$ contains at least one point of $Q$. Since $e_{i+1}$ is a window of $\vis(w_{i+1})$, the point $w_{i+1}$ lies on the line supporting $e_{i+1}$, and its {base} $d_{i+1}$ is a reflex vertex. Extend the segment from $d_{i+1}$ through $w_{i+1}$ until it intersects the boundary of $MV_i(w_{i+1})$ at a point $d_{i+1}'$. Define $l_{i+1} \coloneq  \overline{d_{i+1} d_{i+1}'}$. Because $w_{i+1} \in MV_i(w_{i+1})$ and $l_{i+1}$ passes through $w_{i+1}$, this segment partitions $MV_i(w_{i+1})$ into two sub-polygons $MV_{i}^{(1)}$ and $MV_i^{(2)}$, situated on opposite sides of $l_{i+1}$ (see \cref{recursivestep}). 
By \cref{onlyoneside}, we assume, without loss of generality, that $\vis(e_{v,w_{i+1}})\cap MV_i(w_{i+1})\subseteq MV_i^{(1)}$ for each $v\in \vil_{i+1}(w_{i+1})$ which ensures that $(MV_i^{(2)}\setminus\{d_{i+1}\})\subseteq M_{i+1}$. Therefore, by applying \cref{proofofcontainment}, for any point $v$ in $iN(MV_i^{(2)})$ the collection $\{w_1^*,\ldots, w^*_i,v,w_{i+2},\ldots, w_k\}$ is a witness set of size $k$.

\begin{claim}\label{General CaseexistanceOfgeneratedVertex}
    The boundary of $MV_i^{(2)}$ must contain at least one point from $Q$.
\end{claim}

\ifthenelse{\boolean{shortver}}{}{
\begin{proof}
An edge of $MV_i^{(2)}$ can be of one of the following three types:-

\begin{description}
    \item[Type I.] It is part of $l_{i+1}$.
\item[Type II.] It is part of some {window} of $\vis(w_{i+1})$.
\item[Type III.] It is part of some {wall} of $\vis(w_{i+1})$ or, it is part of some edge of $\vis(e_{w_j^*,w_{i+1}};w_{i+1})$ where $j\leq i$. 
    \end{description}
    

Because there is exactly one Type I edge, no vertex of $MV_i^{(2)}$ can arise from the intersection of two Type I edges. Furthermore, \cref{noconsecutivesecondary} guarantees that windows never appear consecutively, implying that no vertex of $MV_i^{(2)}$ is formed by the intersection of two Type II edges. Now suppose the segment $l_{i+1}$ meets a window $e$ of $\vis(w_{i+1})$. Such an intersection can occur only at a {reflex vertex}, which is the base of $e$; otherwise, $e$ would not be collinear with $w_{i+1}$, since $w_{i+1} \in l_{i+1}$. Consequently, any vertex of $MV_i^{(2)}$ obtained as the intersection of a Type I edge with a Type II edge must be a reflex vertex of the polygon $\pol$. Thus we get only three types of vertices of $MV_i^{(2)}$:

\begin{description}
    \item[Type A.]  A vertex which is an intersection point of two Type III edges or an intersection point of a Type I and a Type II edge (where the latter is always a {reflex vertex} of $\pol$).

     \item[Type B.]  A vertex which is an intersection point of a Type II edge and a Type III edge.
     
    \item[Type C.]  A vertex which is an intersection point of a Type I edge and a Type III edge.
\end{description}


\begin{figure}[ht!]
    \centering
    \includegraphics[width=0.3\linewidth]{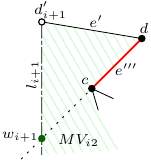}
    \caption{The reflex vertex $c$ is contained in $MV_i^{(2)}$.}
        \label{pfofmvi2}
\end{figure}

Observe that, using \cref{inq}, if a vertex is of Type A then it is in $Q$ (since $w_j^*$ is a clone of some point in $Q$), so we aim to find a point of this type in $MV_i^{(2)}$. Since $d_{i+1}' \in l_{i+1}$, it must be either a Type A or a Type C vertex. If it is Type A, nothing more needs to be shown. Suppose instead that $d_{i+1}'$ is of Type C. In this case, let $e'$ denote the Type III edge whose intersection with $l_{i+1}$ gives rise to $d_{i+1}'$. Let $d$ be the vertex on $e'$ that is adjacent to $d_{i+1}'$. Because $d$ is incident to a Type III edge, it can only be of Type A or Type B. If $d$ is Type A, we are done again. 
If $d$ is of Type B, then there exists a Type II edge $e''$ whose intersection with $e'$ produces the vertex $d$. This edge $e''$ lies on a {window} of $\vis(w_{i+1})$; call this window $e'''$. Let $c$ denote the base of $e'''$. This point $c$ is necessarily a reflex vertex of the polygon $\pol$. See \cref{pfofmvi2} for an illustration. Observe that the triangle $\triangle d_{i+1}' w_{i+1} d$ lies completely inside $MV_i^{(2)}$. Because $e'''$ is collinear with $w_{i+1}$, the point $c$ must lie on the segment $\overline{w_{i+1} d}$. Consequently, $\overline{w_{i+1} d} \subseteq \triangle d_{i+1}' w_{i+1} d \subseteq MV_i^{(2)},$ which shows that $c \in MV_i^{(2)}$. Moreover, since $c$ is a reflex vertex, $c$ is in $Q$. 
\end{proof}}

Using \cref{General CaseexistanceOfgeneratedVertex}, we obtain a point $v_{i+1}\in Q$ contained in $MV_i^{(2)}$. We then substitute $w_{i+1}$ with a clone $Cl(v_{i+1})$ of this point, chosen so that it lies inside $iN(MV_i^{(2)})$. Let this new point be denoted by $w_{i+1}^*$. Since the clone is positioned infinitesimally close to $v_{i+1}$, we may treat $w_{i+1}^*$ as effectively identical to $v_{i+1}$ in all subsequent reasoning. Note that $v_{i+1} \in Q$, and the updated collection $ \{w_1^*, \ldots, w_{i+1}^*, w_{i+2}, \ldots, w_k\}$ is a witness set with cardinality $|W|$.

Thus, at the end of this step, we have obtained a updated witness set $\{w_1^*,\ldots, w_{k-1}^*,w_k\}$ of size $k$ using the principle of mathematical induction.

\subparagraph{Step 2 (Replacing $w_k$):}

Let $w\in W$ be a witness point. $\vis(w)\setminus\bigcup_{v\in \vil(w)}\vis(e_{v,w})$ is disjoint union of \npo s $P_1^w,\ldots, P_{n_w}^w$. Moreover, $w$ lies inside exactly one of those. Without loss of generality, let $w\in P_1^w$. We call the closure of $P_1^w$ as the \emph{movable region} of $w$ and denote it by $MV(w)$. Observe that for any point $v\in iN(MV(w))$, $(W\setminus\{w\})\cup\{v\}$ is a witness set of size $k$ using \cref{proofofcontainment}. Under this framework, we now present the following observation, which plays a key role in this step.
\begin{lemma}\label{removeclonetrick}
   Let $w \in W$ be a witness point such that every other witness point $v \in W \setminus \{w\}$ is already placed at some point of $Q$. Then the polygon $MV(w)$ contains at least two vertices $q_1^w, q_2^w \in Q$ such that $iN(\overline{q_1^w q_2^w}) \subset iN(MV(w))  ~\text{and}~ \overline{q_1^w q_2^w} \subset MV(w).$ \ifthenelse{\boolean{shortver}}{}{In other words, the segment $\overline{q_1^w q_2^w}$ lies entirely inside $MV(w)$, and its interior is contained in the interior of $MV(w)$.}
    \end{lemma}

    \ifthenelse{\boolean{shortver}}{}{
\begin{proof}
    An edge of $MV(w)$ can be of the following types:

    \begin{description}
        \item[Type 1.] It is part of some edge of $\vis(e_{v,w};w)$, $v\in \vil(w)$ or, it is part of some {wall} of $\vis(w)$.

        \item[Type 2.] It is part of some {window} of $\vis(w)$.
    \end{description}

    Since no two {windows} are consecutive (using \cref{noconsecutivesecondary}), there does not exist any vertex of $MV(w)$ which is an intersection point of two {Type $2$} edges. Consequently, the vertices of $MV(w)$ fall into only two possible categories:

    \begin{description}
      \item[Type a.] A vertex which is an intersection point of two {Type $1$} edges.

         \item[Type b.] A vertex which is an intersection point of a {Type $1$} edge and a {Type $2$} edge.
    \end{description}

    Observe that any vertex of {Type a} automatically belongs to $Q$ using \cref{inq}. Thus, if all vertices of $MV(w)$ are of Type a, the claim follows immediately. Otherwise, suppose that $MV(w)$ contains a vertex of {Type b}; denote this vertex by $b$.

    \begin{figure}[ht!]
        \centering
        \includegraphics[width=0.2\linewidth]{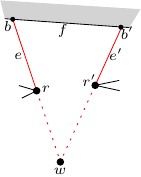}
        \caption{$iN(\overline{q_1^wq_2^w})\subset iN(MV(w))$ and $\overline{q_1^wq_2^w}\subset MV(w)$, where $q_1^w=r,$ $q_2^w$ can be $b'$ or, $r'$.}
        \label{wktovetr}
    \end{figure}
    
    Since $b$ is of {Type b}, it is the intersection point of a window $e$ of $\vis(w)$ and an edge $f$ of $MV(w)$ of {Type a}. Let $b'$ denote the other endpoint of $f$, and let $r$ be the {base} of the window $e$. Observe that the triangle $\triangle bwb'$ lies entirely within $MV(w)$. Because $e$ is collinear with $w$, we have $ r \in \overline{bw} \subset \triangle bwb' \subset MV(w). $ We now consider two cases.

    \begin{itemize}
        \item \textbf{Case I ($b'$ is of {Type a}):} In this case we take $q_1^w=r$ and $q_2^w=b'$ as $r$ being the {base} of a {window} is a {reflex vertex}. And $iN(\overline{q_1^wq_2^w})\subset iN(MV(w))$ and $\overline{q_1^wq_2^w}\subset MV(w)$ follow from the fact that $\triangle bwb'\subset MV(w)$.

        \smallskip 

        \item     \textbf{Case II ($b'$ is of {Type b}):} Let $b'$ be the intersection point of the segment $f$ with a {window} $e'$ of $\vis(w)$. Denote by $r'$ the {base} of $e'$. Since $e'$ is collinear with $w$, we have $r' \in \overline{wb'} \subset \triangle bwb' \subset MV(w).$ In this case, we choose $q_1^w = r$ and $q_2^w = r'$, noting that the points $r$ and $r'$, being the bases of windows, are reflex vertices. Furthermore, $iN(\overline{q_1^w q_2^w}) \subset iN(MV(w))\text{ and }\,\overline{q_1^w q_2^w} \subset MV(w), $ which follow from the fact that $\triangle bwb' \subset MV(w)$. See \cref{wktovetr} for an illustration.
    \end{itemize}



\end{proof}
}

By Step $1$, each witness point $w_j,  j \in \{1,\ldots,k-1\}$ has been replaced by $w_j^*$. Each $w_j^*$ is a clone of some point of $Q$, meaning that its position is infinitesimally close to a point of $Q$. Therefore, the setup of \cref{removeclonetrick} is applicable to $w_k$. Using \cref{removeclonetrick}, we obtain two points $q_1^{w_k}, q_2^{w_k}\in Q$ such that $ iN(\overline{q_1^{w_k} q_2^{w_k}}) \subset iN(MV(w_k)) ~ \text{and} ~ \overline{q_1^{w_k} q_2^{w_k}} \subset MV(w_k). $ Let $w_k'$ denote the midpoint of  $\overline{q_1^{w_k} q_2^{w_k}}$. Now $w_k'$ lies in the interior of $MV(w_k)$ and belongs to $Q$ using \cref{inq}. Thus, we obtain an updated witness set $\{w_1^*, \ldots, w_{k-1}^*, w_k'\}$ of size $k$, where both $w_1^*$ and $w_k'$ lie in $Q$, and for each $2 \le j \le k-1$,  $w_j^*$ is a clone of a point in $Q$.

\subparagraph{Step 3 (Dealing  with Clones):} 

Since $w_1^* = w_1$ is an original vertex of $\pol$, we simply set $w_1' \coloneq w_1$ and the same technique used for $w_k$ applies to every other witness point. For any $j \in \{2,\ldots,k-1\}$, every neighbor witness of $w_j^*$ is either a point of $Q$ or a clone of some point of $Q$. Hence, the setup of \cref{removeclonetrick} is valid for $w_j^*$. Applying \cref{removeclonetrick}, we obtain two points $q_1^{w_j^*},q_2^{w_j^*}\in Q$ such that $iN(\overline{q_1^{w_j^*} q_2^{w_j^*}})\subset iN(MV(w_j^*))$ and  $\overline{q_1^{w_j^*} q_2^{w_j^*}}$ lies within $MV(w_j^*)$. Let ${w_j}'$ denote the midpoint of $\overline{q_1^{w_j^*} q_2^{w_j^*}}$. Clearly, ${w_j}'$ lies in the interior of $MV(w_j^*)$ and belongs to $Q$ using \cref{inq}. We now replace $w_j^*$ with ${w_j}'$.


Of course, we can do this process by following any order of the set $\{2,\ldots,k-1\}$ but for convenience we do this replacing by some point of $Q$ process by the order $(k-1,\ldots, 2)$.  We begin by replacing $w_{k-1}^*$ with $w_{k-1}'$, after which the witness set becomes $\{w_1',\, w_2^*,\, \ldots,\, w_{k-2}^*,\, w_{k-1}',\, w_k'\}$. Next, we replace $w_{k-2}^*$ with $w_{k-2}'$ and update the witness set to $\{w_1',\, w_2^*,\, \ldots,\, w_{k-3}^*,\, w_{k-2}',\, w_{k-1}',\, w_k'\}$. Proceeding in this way, we finally replace $w_2^*$ with $w_2'$, resulting in the witness set $\{w_1',\, \ldots,\, w_k'\}$. Thus, each {witness point} has been replaced by a point of $Q$. This leads to the following theorem.

\begin{theorem}\label{mainthm}
    Let $W=\{w_1,\ldots,w_k\}$ be a {witness set} in a simple polygon $\pol$. Then there exist  $w_1',\ldots,w_k' \in Q$ such that  $\{w_1',\ldots,w_k'\}$ is also a {witness set} of cardinality $|W|$.
\end{theorem}

\subsubsection{Running Time analysis and Size of \texorpdfstring{$Q$}{Q}} \label{sec:Q}
 Let $n:=|V(\pol)|=|E(\pol)|$ and $r:=|R|$.

\subparagraph{Size of $Q$:}  We give a concise asymptotic bound on the cardinality of the set $Q$ produced by Algorithm~\ref{algo_1}.
The algorithm has three phases: a first refinement loop of $2(k-1)$
iterations, a midpoint-squaring step, and a second refinement loop of
$2(k-2)$ iterations. A brief  argument  yields the following estimates.

\begin{itemize}
    \item After the first refinement loop one obtains
$ |Q^{(1)}| = \mathcal{O}(n^{\,2k-1} r^{\,2k-2})
$

\item The midpoint insertion step can (in the worst case) square this quantity,
so immediately afterwards we have
$
q_0 =\mathcal{O} ((|Q^{(1)}|)^2) =  \mathcal{O}(n^{\,4k-2} r^{\,4k-4}).
$

\item In the second refinement loop the sizes satisfy a super-multiplicative
recurrence of the form
$
q_{i} \;=\; \mathcal{O}\bigl(q_{i-1}^4 r^4\bigr),
$
so that after $t=2(k-2)$ iterations one obtains
$
q_t =  \mathcal{O} ( q_0^{\,4^{t}}  r^{\,4t}).
$
\end{itemize}

Substituting $q_0 = \mathcal{O}(n^{4k-2} r^{4k-4})$ and $t=2(k-2)$ and simplifying we obtain

$|Q| \;=\; n^{(4k-2)\,4^{2k-4}} \cdot r^{\,\mathcal{O}(4^{2k-4}\,k)} .
$

\subparagraph{Running time} From the construction of Algorithm~\ref{algo_1}, it follows  that running time is $\mathcal{O}(\lvert Q\rvert)$. 

Hence we conclude the following result.

\thmDisFPT*

\subsection{Putting Everything Together: Solving {\sc Witness Set}}
\label{sec:witness}

We now show how the discretization framework developed earlier leads to an exact algorithm for the {\sc Witness Set} problem in simple polygons.

\subparagraph{From continuous to discrete.} Given the polygon $\pol$ and an integer $k$, we apply 
Algorithm~\ref{algo_1} to the instance $(\pol,k)$, which constructs a finite point 
set $Q\subseteq \pol$.  
By \cref{mainthm}, every witness set
$W$ of size $k$ admits an equivalent witness set $W'\subseteq Q$ of the same
size.  
Thus the continuous problem reduces to the following question: {\em Among all subsets of $Q$, find one of maximum cardinality whose 
pairwise visibility regions do not overlap}.

\subparagraph{Solving the discrete instance.}
Let $n = |V(\pol)|$ and let $|Q|$ denote the size of the discretization produced by
Algorithm~\ref{algo_1}.   By the analysis in \Cref{sec:Q}, we have
$
|Q| \;=\; n^{(4k-2)\,4^{2k-4}} \cdot r^{\,\mathcal{O}(4^{2k-4}\,k)}.
$
Since $r\le n$, this simplifies to
$
|Q| \;=\; n^{\,\mathcal{O}(4^{2k})}.
$

The discrete version of {\sc Witness Set} (where one must choose witnesses only
from $Q$) can be solved in
$
\mathcal{O}(|Q|^{3}\cdot |V(\pol)|^{3})
$
time using Theorem~1.1 of~\cite{das2025witnesssetmonotonepolygons}.
Applying this  to $Q$ yields a maximum witness set contained entirely in
$Q$, which by our replacement theorem corresponds to a witness set of size at least $k$
in polygon $\pol$. Thus we have the following theorem.

\thmfptalgorithm*

\section{{\sf NP}-hardness of {\sc Discrete Witness Set}}\label{sec:hard}

In this section, we deal with polygon with holes.

\begin{tcolorbox}[enhanced,
    title={\color{black}{\sc Discrete Witness Set}},
    colback=white,
    boxrule=0.4pt,
    attach boxed title to top center={xshift=-4.5cm, yshift*=-2.5mm},
    boxed title style={size=small, frame hidden, colback=white}]
    \vspace{-1mm}
    \textbf{Input:} A  polygon $\mathscr{P}$, a finite set of points $Q \subset \mathscr{P}$, and a positive integer $k$.\\
    \textbf{Question:} Does there exist a point set $S\subseteq Q$ with $|S|\ge k$
such that every point of $\pol$ is visible from at most one point of $S$?
\end{tcolorbox}


\thmnph*

\noindent We present a polynomial-time reduction from the \textsc{Planar Monotone Rectilinear 3SAT (PMR3SAT)} problem, which is known to be \textsf{NP}-complete~\cite{DBLP:journals/ijcga/BergK12}. 
\ifthenelse{\boolean{shortver}}{}{An instance of PMR3SAT is a CNF formula \(\Phi\) with \(n\) variables and \(m\) clauses such that:
\begin{itemize}
  \item each clause  of $\Phi$ has at most three literals,
  \item each clause is monotone (either all-positive or all-negative),
  \item the variable–clause incidence graph admits a rectilinear planar embedding where variables are represented by axis-aligned rectangles lying on a common horizontal line, positive clauses are above that line, negative clauses below, and each clause is connected to its variables by vertical segments without crossings.
\end{itemize}

The objective of \textsc{PMR3SAT} is to determine whether $\Phi$ is satisfiable,  that is, whether there exists a truth assignment to its variables such that every clause contains at least one literal that evaluates to \texttt{true}.}

\begin{figure}[ht!]
    \centering
    \includegraphics[width=1\linewidth]{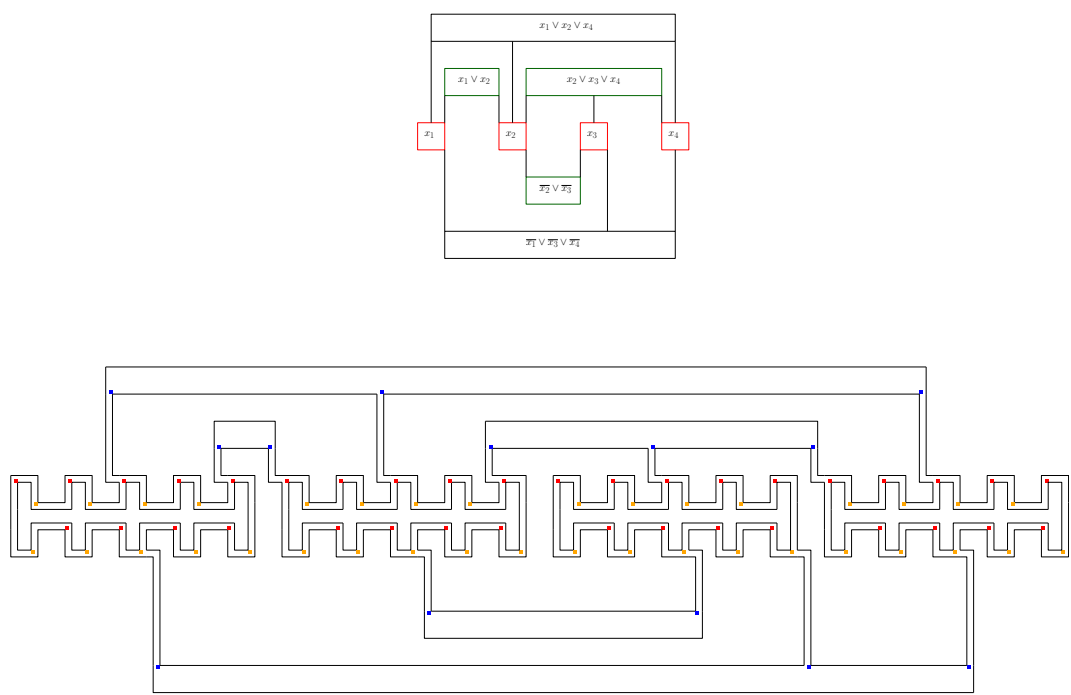}
    \caption{An illustration of the   reduction from {\sc PMR3SAT} to {\sc Discrete Witness Set}.}
    \label{wit1}
\end{figure}

\subparagraph{Brief Overview of the reduction.}
Given a rectilinear embedding of an instance \(\Phi\) of PMR3SAT (witn $n$ variable and $m$ clauses) we build a polygon \(\pol\) with holes and a candidate point set \(Q\subset \pol\) with the following properties (see Figure~\ref{wit1} for an illustration): each variable \(v\) corresponds to a \emph{variable gadget} \(R_v\) placed on the variable line; each clause \(C\) corresponds to a \emph{clause gadget} \(R_C\) placed either above (if \(C\) is positive) or below (if \(C\) is negative) the variable line; and whenever variable \(v\) occurs in clause \(C\) we connect \(R_v\) and \(R_C\) by a thin vertical corridor (a narrow rectilinear strip) that enforces specific visibility relations between points in the two gadgets. We choose the candidate set \(Q\) so that:
\begin{itemize}
  \item every clause gadget \(R_C\) contains exactly three \emph{blue} candidate points (one per incident variable edge position),
  \item every variable gadget \(R_v\) contains \(m\) \emph{orange} points and \(m\) \emph{red} points (for a uniform counting argument we place \(m\) copies of each color per variable).
\end{itemize}
Thus \(|Q| = 2mn + 3m\).
Intuitively each variable gadget allows us to pick either all orange points (representing truth) or all red points (representing false), but not both. Each blue point in a clause gadget can be selected only if the connected variable gadget has chosen the matching color (so a positive clause's blue can be selected only if the corresponding variable gadget chose orange; a negative clause's blue only if that variable chose red). Clause gadgets are arranged so that at most one blue point per clause can be selected (blue points in the same clause conflict pairwise). We set the target threshold \(k := mn + m\). The reduction ensures \(\Phi\) satisfiable if and only if  there exists an witness set \(S\subseteq Q\) with \(|S|\ge k\).

\ifthenelse{\boolean{shortver}}{}{

 \paragraph*{Gadget Construction.}

 \begin{description}
        \item[Variable Gadget.] Each variable gadget \(R_v\) is an rectilinear polygon   placed on the variable line. Inside \(R_v\) we place two groups of candidate points: \(m\) red points arranged along the top part of \(R_v\), \(m\) orange points arranged along the bottom part of \(R_v\). We placed the red-orange points in such a way that in each variable gadget \(R_v\) the following hold:
\begin{enumerate}
  \item Any two orange (resp, red) points have disjoint visibility regions inside \(R_v\) (so they can all be simultaneously chosen).

  \item Any orange point and the next  red point (according to there position along the boundary) have intersecting visibility regions, so we cannot choose both of them.
\end{enumerate}

        \item[Clause Gadget.] Each clause gadget \(R_C\) is axis-aligned rectangle  placed above (for positive) or below (for negative) the variable line. In \(R_C\) we place three blue points corresponding to the three incident variable-edges (if a clause has fewer than three literals we leave unused positions empty). We placed the blue points in such a way that in each clause gadget \(R_C\) the following hold:

\begin{enumerate}
  \item Any two blue points in the same clause gadget have intersecting visibility regions, so at most one blue point can be chosen from \(R_C\)).
  \item For each incident variable \(v\) of \(C\), the blue point placed at the corridor endpoint toward \(v\) has its visibility region intersecting exactly the opposite-color points in \(R_v\) (i.e., for a positive clause the blue point intersects red points' visibility but not orange points', and vice versa for negative clauses). This ensures consistency between clauses and variables as described below.
\end{enumerate}

        \item[Connection between variable and clause gadget.]
        For each incidence (edge) between variable \(v\) and clause \(C\) in the rectilinear embedding we create a narrow vertical corridor connecting the corresponding positions in \(R_v\) and \(R_C\). 
\begin{itemize}
  \item A blue point in \(R_C\) can see into \(R_v\) only the single pocket (the corresponding red or orange point) assigned for that corridor; it cannot see points elsewhere.
\end{itemize}

\item[Target Value] We construct a polygon \(\pol\) with holes and set of candidate points \(Q\) with $|Q| = 2mn + 3m$, and set the target $k := mn + m$.
\end{description}

The polygon can be obtained in time polynomial in the size of the PMR3SAT embedding (the embedding is planar and rectilinear). Next we prove the correctness of the reduction.

\paragraph*{Correctness of the reduction}
We prove that \(\Phi\) is  satisfiable if and only if  there exists an witness set  \(S\subseteq Q\) with \(|S|\ge k\).

\begin{claim}[Completeness]
    If $\Phi$ is satisfiable then $\pol$ has $k=mn+m$ witnesses \(S\subseteq Q\).
\end{claim}

\begin{claimproof}
    We construct \(S\) as follows:
\begin{itemize}
  \item For each variable \(v\): if \(v\) is assigned \texttt{true}, add all \(m\) orange points of \(R_v\) to \(S\); if \(v\) is assigned \texttt{false}, add all \(m\) red points of \(R_v\) to \(S\). 
  
  \item For each clause \(C\): since the assignment satisfies \(C\), there exists a literal in \(C\) set to \texttt{true}. If \(C\) is positive there exists a variable \(v\in C\) with \(v=\texttt{true}\). Consider the blue point in \(R_C\) positioned toward \(v\); by construction that blue point does \emph{not} conflict with the orange points in \(R_v\) but would conflict with red points. Since we chose orange in \(R_v\), the blue point is compatible with all previously chosen points. Similarly for negative clauses and red choices.
\end{itemize}
Thus we choose \(mn\) points from variable gadgets (each of the \(n\) variables contributes \(m\) same-color points) and \(m\) blue points (one per clause). Total \(|S| = mn + m = k\). Pairwise disjointness of visibility regions follows from the gadget invariants described above. This proves completeness.
\end{claimproof}

\begin{claim}[Soundness]
     If $\pol$ has $k=mn+m$ witnesses \(S\subseteq Q\) then $\Phi$ is satisfiable.
\end{claim}

\begin{claimproof}
    First observe that  each clause gadget contributes at most one blue point to any valid witness set. Hence the total number of blue points in \(S\) is at most \(m\). Second, within each variable gadget \(R_v\),  we cannot simultaneously choose an orange and a red point. But  we can choose up to \(m\) points in \(R_v\) by selecting all orange or all red assigned copies.

Hence the maximum number of points one can choose while respecting pairwise-disjointness is achieved by, for each variable gadget, choosing all \(m\) points of one color (either all orange or all red), and for each clause choosing at most one blue; this yields an upper bound of \(mn + m\) points total. Since \(|S|\ge mn+m\), it follows that:
\begin{itemize}
  \item For every variable gadget \(R_v\) the set \(S\) must contain exactly \(m\) chosen points from \(R_v\), and these must all be of the same color (otherwise the variable gadget could contribute strictly fewer than \(m\) points because red and orange conflict pairwise).
  \item For every clause gadget, \(S\) must contain exactly one blue point (because if some clause contributed zero blue points, then even choosing \(m\) from each variable would give at most \(mn < mn+m\) total).
\end{itemize}

Thus from \(S\) we can extract a truth assignment: set variable \(v\) to \texttt{true} if \(S\) selects all orange points of \(R_v\), and set \(v\) to \texttt{false} if \(S\) selects all red points of \(R_v\).

It remains to verify that every clause \(C\) is satisfied by this assignment. Consider the blue point of \(R_C\) that is contained in \(S\). By construction of the corridors, that blue point is compatible with the chosen color in the variable gadget at the corridor it sits above/below; specifically:
\begin{itemize}
  \item If \(C\) is a positive clause and the chosen blue point in \(R_C\) is the one adjacent to variable \(v\), then that blue point is compatible only when \(R_v\) contributed orange points to \(S\) (i.e., \(v=\texttt{true}\)). Therefore this clause has a true literal.
  \item If \(C\) is negative, the analogous argument shows the adjacent variable \(v\) contributed red points to \(S\) (i.e., \(v=\texttt{false}\)), so the clause is satisfied.
\end{itemize}
Hence every clause contains at least one true literal under the derived assignment, and therefore \(\Phi\) is satisfiable. This proves soundness.
\end{claimproof}

Putting together we can conclude that Discrete Witness Set is NP-hard. Putting together we can conclude that  \textsc{Discrete Witness Set}  is {\sf NP}-hard.  It is easy
to observe that a candidate set $S$ can be verified in polynomial time (similar argument as \Cref{sec:np}).

Hence we have the  \Cref{theo-hard1}.
}

\section{{\sc Witness Set} is in {\sf NP}}\label{sec:np}

Proving $\mathsf{NP}$-hardness for visibility problems requires care, as the classical {\sc Art Gallery} problem is $\exists\mathbb{R}$-complete~\cite{DBLP:journals/jacm/AbrahamsenAM22}, and therefore not in $\mathsf{NP}$ unless $\mathsf{NP} = \exists\mathbb{R}$. In contrast, we show that its dual problem, the {\sc Witness Set} problem, lies in $\mathsf{NP}$. A certificate consists of a finite set of points $S = \{p_1, \dots, p_k\}$ inside a  polygon $\pol$. To verify that $S$ is a valid witness set, we proceed as follows.

\begin{enumerate}
\item For each point $p_i \in S$, compute its visibility polygon $\mathrm{Vis}(p_i)$ within $\pol$. It is well known that the visibility polygon of a point in a  polygon with $n$ vertices can be computed in linear time~\cite{DBLP:journals/cvgip/Lee83}.

\item For every pair of distinct points $p_i, p_j$, check whether the corresponding visibility polygons $\mathrm{Vis}(p_i)$ and $\mathrm{Vis}(p_j)$ intersect. This can be done by testing for intersection between two polygons, which reduces to checking for edge intersections and containment, and can be solved in polynomial time~\cite{DBLP:conf/focs/ShamosH76}.

\item Accept if and only if for all pairs $p_i, p_j$ with $i \neq j$, we have
$
\mathrm{Vis}(p_i) \cap \mathrm{Vis}(p_j) = \emptyset.
$
\end{enumerate}

\section{Conclusion}\label{sec-conclusion}

In this work, we present first  discretization framework for the
{\sc Witness Set}  in simple polygons, together with  structural
characterizations that yield an $n^{f(k)}$–time algorithm to solve {\sc Witness Set} with some  $f$, where $k$ is the desired witness-set size.  Several directions for future research remain open. 

\begin{itemize}

    \item Our discretization produces a finite candidate set
$Q$ of size $n^{h(k)}$ for computable function $h$, the current bounds are unlikely to be tight.
Obtaining substantially smaller discretizations -- ideally of size
$g(r,k)\,n^{\mathcal{O}(1)}$ would significantly strengthen the approach.

\item Determining whether constant-factor approximations or even
PTAS-type algorithms exist for {\sc Witness Set}  is an appealing  question.

\item While our exact algorithm places the problem in {\sf XP} when
parameterized by $k$, it remains open whether an {\sf FPT} algorithm in $k$ is possible.

\item Our {\sf NP}-hardness result is on  \textsc{Discrete Witness Set} for the polygons with holes.  It remains an open question whether the \textsc{Witness Set} problem is {\sf NP}-hard.
\end{itemize}

\bibliography{cite}
\end{document}